\documentclass[aps,twocolumn,pra,superscriptaddress]{revtex4-1}

\usepackage[toc,page]{appendix}
\usepackage{graphicx}
\usepackage{lipsum}
\usepackage{amsmath}
\usepackage{amsfonts}
\usepackage{amsthm}
\usepackage{mathtools}
\usepackage{amssymb}
\usepackage{multirow}
\usepackage{array}
\usepackage{bm}
\usepackage{natbib}
\usepackage{diagbox}

\newtheorem{theorem}{Theorem}

\newtheorem{definition}{Definition}
\newtheorem{lemma}{Lemma}
\newtheorem{corollary}{Corollary}
\newtheorem{remark}{Remark}

\newcommand{\RD}{R\'enyi divergence }
\newcommand{\RE}{R\'enyi entropy }
\newcommand{\dv}{\Vert}
\newcommand{\D}[3]{D_{#1}(#2 \dv #3)}
\newcommand{\SD}[3]{D_{#1}(#2 \dv #3)}
\newcommand{\spn}[1]{\operatorname{span} \left\{ #1 \right\}}
\newcommand{\supp}[1]{\operatorname{supp}\left(#1\right)}
\newcommand{\sgn}[1]{\operatorname{sgn}\left(#1\right)}

\newcommand{\ket}[1]{\vert#1 \rangle}
\newcommand{\outprod}[2]{\vert #1 \rangle \langle #2 \vert}

\newcommand{\Tr}{\operatorname{Tr}}
\newcommand{\rank}[1]{\operatorname{rank}\left( #1 \right)}

\DeclarePairedDelimiter{\ceil}{\lceil}{\rceil}

\usepackage[colorlinks=true,linkcolor=blue,citecolor=red,plainpages=false,pdfpagelabels]%
{hyperref}%

\newcolumntype{L}{>{\centering\arraybackslash}m{2.6cm}}

\newcolumntype{M}{>{\centering\arraybackslash}m{3.75cm}}

\allowdisplaybreaks

\begin{document}

\title{Relative Entropy and Catalytic Relative Majorization}

\author{Soorya Rethinasamy}

\affiliation{Birla Institute of Technology and Science, Pilani, Rajasthan 333031, India}
\affiliation{Hearne Institute for Theoretical Physics, Department of Physics and Astronomy, and Center for Computation and Technology,
Louisiana State University, Baton Rouge, Louisiana 70803, USA}

\author{Mark M. Wilde}

\affiliation{Hearne Institute for Theoretical Physics, Department of Physics and Astronomy, and Center for Computation and Technology,
Louisiana State University, Baton Rouge, Louisiana 70803, USA}

\affiliation{Stanford Institute for Theoretical Physics, 
Stanford University, Stanford, California 94305, USA}

\begin{abstract}
Given two pairs of quantum states, a fundamental question in the resource theory of asymmetric  distinguishability is to determine whether there exists a quantum channel converting one pair to the other. In this work, we reframe this question in such a way that a catalyst can be used to help perform the transformation, with the only constraint on the catalyst being  that its reduced state is returned unchanged, so that it can be used again to assist a future transformation. What we find here, for the special case in which the states in a given pair are commuting, and thus quasi-classical, is that this catalytic transformation can be performed if and only if the relative entropy of one pair of states is larger than that of the other pair. This result endows the relative entropy with a fundamental operational meaning that goes beyond its traditional interpretation in the setting of independent and identical resources. Our finding thus has an immediate application and interpretation in the resource theory of asymmetric distinguishability, and we expect it to find application in other domains.
\end{abstract}

\date{\today}

\maketitle

\section{Introduction}

The majorization partial ordering has been studied in light of the following question \cite{G.H.Hardy1952}: What does it mean for one probability distribution to be more disordered than another? While there exist several equivalent characterizations of majorization, the most fundamental is arguably the notion that a distribution $p$ majorizes another distribution $p'$ when $p'$ can be obtained from $p$ by the action of a doubly stochastic channel. This captures the intuition that $p'$ is more disordered than $p$.

The theory of majorization has several widespread applications, some of which are in quantum resource theories \cite{Chitambar_2019}. For example,   transformations of pure bipartite states by means of local operations and classical communication can be analyzed in terms of  majorization  of the Schmidt coefficients of the states \cite{Nielsen1999}. Majorization has been shown to determine whether state transformations are possible not just in the resource theory of entanglement, but also in coherence \cite{DBG15, Zhu_2017, Winter_2016} and purity \cite{Horodecki_2003} as well. 

As it turns out, majorization is a special case of a more general concept introduced in \cite{RM76,RSS78,RSS80}, which is now called relative majorization \cite{Renes_2016,Buscemi_2017}.
This more general notion also has a rich history in the context of statistics, where it goes by the name of ``statistical comparison'' or ``comparison of experiments'' \cite{Blackwell1953}. Relative majorization is further generalized by the concept of matrix majorization, as presented in \cite{Dahl99}, while the $d$-majorization of \cite{V71} is a special case of relative majorization.

To recall the notion of relative majorization, a pair $(p , q)$ of probability distributions $p$ and $q$ relatively majorizes another pair $(p' , q')$ of probability distributions $p'$ and $q'$  if there exists a classical channel (conditional probability distribution) that converts $p$ to $p'$ and $q$ to $q'$. That is, $(p , q)$ relatively majorizes  $(p' , q')$, denoted by
\begin{equation}
(p,q) \succ (p',q'),   \label{eq:rel-maj-notation} 
\end{equation}
if there exists a classical channel or stochastic matrix  $N$ such that $p' = N p$ and $q' = N q$. Deciding relative majorization can be accomplished by means of a linear program \cite{Dahl99} that is efficient in the alphabet sizes of the distributions.  In each pair $(p , q)$ and $(p' , q')$, if the second distributions $q$ and $q'$ are equal, then relative majorization collapses to $d$-majorization. If they are furthermore set to the uniform distribution, then $d$-majorization collapses to majorization, demonstrating that relative majorization is indeed a generalization of majorization. 

\begin{figure}
\begin{centering}
\includegraphics[width=0.48\textwidth]{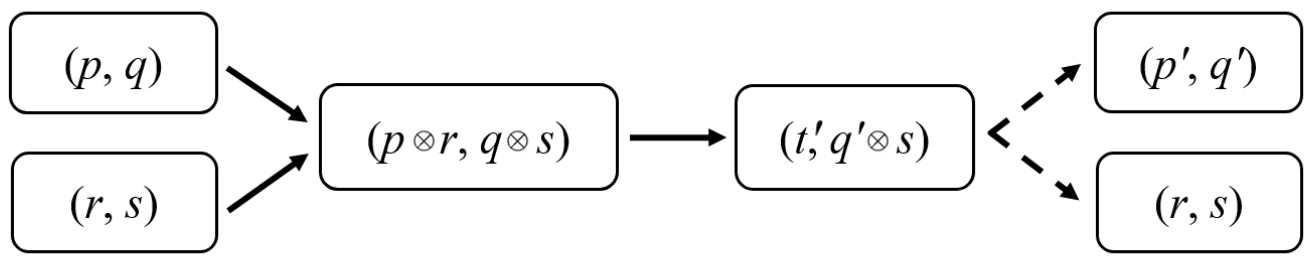}
\caption{$(p, q)$ is a pair of probability distributions to be converted to another pair $(p', q')$, by means of a classical channel and the use of the catalyst pair $(r, s)$. The distribution $t'$ after transformation has first marginal $p'$ and second marginal~$r$. The dashed arrows on the right indicate marginalization of the joint probability distributions $t'$ and $q' \otimes s$.}
\label{fig:flowchart}
\end{centering}
\end{figure}

Another generalization of majorization comes in the form of catalysis \cite{Jonathan1999}. A catalyst is an ancillary distribution whose presence enables certain transformations that would otherwise be impossible, with the only constraint being that its reduced state is returned unchanged. Catalysis has found application in various resource theories, including entanglement \cite{Jonathan1999, Daftuar2001}, thermodynamics \cite{Brand_o_2015}, purity \cite{Boes_2019}, and coherence \cite{BSW16, CZZ+19}.

In this paper, we reformulate the transformation task in relative majorization by allowing for a catalyst that can aid the transformation. In a general form for this task, a catalyst consists of a pair $(r, s)$ of probability distributions $r$ and $s$. The pair can be used in conjunction with the original input pair $(p , q)$ to generate the output pair $(p', q')$ by means of a classical channel $\Lambda$ acting on $\left(p \otimes r, q \otimes s\right)$. To make the catalytic task non-trivial, we demand that the first and second marginals of $\Lambda(p \otimes r)$ are $p'$ and $r$, respectively. We also demand for $\Lambda(q \otimes s)$ to be exactly equal to $q' \otimes s$. As a consequence of this demand, the catalyst is returned unchanged and can be used in future catalytic tasks for distributions that are independent of $p$ and $p'$ \footnote{The stated assumption about reuse of the catalyst is quite important, and violation of it could lead to problems in applications. To be clear, once the transformation from the distribution pair $(p,q)$ to the distribution pair $(p',q') $ is completed, it is necessary for an incoming distribution pair to be independent of $(p,q)$ in order to reuse the catalyst. If it is not, the conditions of our theorem do not hold and the new transformation is not guaranteed to succeed. See \cite{M_ller_2018} for further discussions.}. We call this task \textit{relative majorization assisted by a catalyst with correlations}. This task is succinctly summarized in Figure~\ref{fig:flowchart}.

We now summarize one result of our paper. We note that we focus on a special case of the task mentioned above, in which the second distribution $s$ of the catalyst pair is fixed to be the uniform distribution, denoted by $\eta$. Let $p$, $q$, $p'$, and $q'$ be probability distributions over the same alphabet. Suppose that $q$ and $q'$ have rational entries and full support. Define the \textit{relative spectrum} of the pair $(p,q)$ of probability distributions $p$ and $q$ to be the set formed from the distinct entries of the following set:
\begin{equation}
    \{p(x) / q(x) \}_x.
\end{equation}
Suppose further that the relative spectra of the pairs $(p,q)$ and $(p',q')$ are different.
  Then, under these assumptions on $p$, $q$, $p'$, and $q'$, it is possible to transform the pair $(p, q)$ to the pair $(p', q')$ in the above sense if and only if 
\begin{equation}
D(p \Vert q) > D(p' \Vert q') \text{ and } D_0(p\Vert q) \geq D_0(p'\Vert q')
\label{eq:summary-one-result}
\end{equation}
where the relative entropy $D(p \Vert q)$ is defined as
\begin{equation}
D(p \Vert q) := \sum_x p(x) \ln \!\left(\frac{p(x)}{q(x)}\right),
\end{equation}
and the min-relative entropy $D_0(p \Vert q)$ is defined as
\begin{equation}
D_0(p \Vert q) := -\ln \left[\sum_{x : p(x) \neq 0} q(x)\right].
\end{equation}
Observe that $D_0(p \Vert q) \geq 0$ for all distributions $p$ and $q$ and $D_0(p \Vert q) = 0$ if $p$ has full support, so that the min-relative entropy constraint in \eqref{eq:summary-one-result} can be interpreted as a technical condition.

Another contribution of our paper is that it is possible to perform an inexact, yet arbitrarily accurate transformation with fewer assumptions on $p$, $q$, $p'$, and $q'$, but again under the assumption that $s$ is the uniform distribution.

We establish these results by building on the prior work of \cite{Brand_o_2015, M_ller_2018}. As done in the prior works, we employ an embedding channel as a technical tool \cite{Brand_o_2015} (see also \cite[page 227]{RSS80}), and we also allow for the catalyst to become correlated with the target distribution~\cite{M_ller_2018}. 

We note here that our results apply to the more general ``quasi-classical'' case discussed in the abstract, in which the first pair consists of commuting quantum states $\rho$ and $\sigma$ and the second pair consists of commuting quantum states $\rho'$ and $\sigma'$. This follows as a direct consequence of the fact that commuting quantum states are diagonal in the same basis and thus effectively classical, along with the fact that there is a local unitary channel that takes the common basis of the first pair to the common basis of the second pair.

\section{Main result}

\label{sec:MainRes}

We begin by stating the main results of our paper and an associated corollary. We present all proofs in Appendices~\ref{sec:proofThm1}, \ref{sec:proofThm2}, and \ref{sec:proofCor1}, while Appendices~\ref{sec:Notat} and \ref{sec:TechLem} review some basics needed to prove the main results.

\begin{theorem}[One-Shot Characterization of  Exact Pair Transformations]
\label{Theorem1}
Let $p, q, p', q'$ be probability distributions  on the same alphabet, such that the relative spectra of the pairs $(p,q)$ and $(p',q')$ are different, and $q$ and $q'$ have full support and only rational entries. Then the following conditions are equivalent:
\begin{enumerate}
	\item $\D{}{p}{q} > \D{}{p'}{q'}$ and $\D{0}{p}{q} \geq \D{0}{p'}{q'}$
	\item For all $\gamma>0$, there exists a probability distribution~$r$, a joint distribution $t'$, and a classical channel $\Lambda$ such that
	\begin{enumerate}
		\item $\Lambda(p \otimes r) = t'$ with marginals $p'$ and $r$
		\item $\Lambda(q \otimes \eta) = q' \otimes \eta$
		\item $\D{}{t'}{p' \otimes r} < \gamma$
	\end{enumerate}
	where $\eta$ is the uniform distribution on the support of $r$.
\end{enumerate}
\end{theorem}

We note that statement 2 above can alternatively be written in the notation of \eqref{eq:rel-maj-notation} as follows: For all $\gamma>0$, there exists a probability distribution~$r$ and a joint distribution $t'$ such that
	\begin{align}
		(p \otimes r, q \otimes \eta) & \succ  (t',q' \otimes \eta), \\
		\D{}{t'}{p' \otimes r} & < \gamma,
	\end{align}
	where   $t'$ has first marginal $p'$ and second marginal $r$, and $\eta$ is the uniform distribution on the support of $r$.

Theorem~\ref{Theorem1} indicates that the relative entropy is the main relevant information quantity that characterizes pair transformations when catalysts are available, other than the min-relative entropy as an additional technical condition in the case that $p'$ does not have full support. Furthermore, since $\gamma >0$ is arbitrary, the relative entropy $\D{}{t'}{p' \otimes r}$ can be made as small as desired and thus the output of the classical channel $\Lambda$ arbitrarily close to the product of $p'$ and $r$.

The constraint that the relative spectra are different, that $q$ and $q'$ are rational, and the min-relative entropy inequality can be relaxed by allowing for a slight error in the formation of the target state, as follows:

\begin{theorem}[One-Shot Characterization of Approximate Pair Transformations]
\label{Theorem2}
Given probability distributions $p, q, p', q'$ on the same alphabet, such that  $q$ and $q'$ have full support, the following conditions are equivalent:

\begin{enumerate}
	\item $\D{}{p}{q} \geq \D{}{p'}{q'}$
	\item For all $\varepsilon \in (0,1)$ and $\gamma>0$, there exist probability distributions $r$ and $p'_\varepsilon$, a joint distribution $t'_\varepsilon$, and a classical channel $\Lambda$ such that
	\begin{enumerate}
		\item $\Lambda(p \otimes r) = t'_\varepsilon$ with marginals $p'_\varepsilon$ and $r$,
		\item $\Lambda(q \otimes \eta) = q' \otimes \eta$,
		\item  $\frac{1}{2}\left\Vert p' - p'_\varepsilon \right\Vert_1 \leq \varepsilon$,
		\item $\D{}{t'_\varepsilon}{p'_\varepsilon \otimes r} < \gamma$,
\end{enumerate}
where $\left\Vert w - v\right\Vert_1 := \sum_i |w_i - v_i|$ 
	and  $\eta$ is the uniform distribution on the support of $r$.
\end{enumerate}
\end{theorem}

We again note that statement 2 above can alternatively be written in the notation of \eqref{eq:rel-maj-notation} as follows: For all $\varepsilon \in (0,1)$ and $\gamma>0$, there exist probability distributions $r$ and $p'_\varepsilon$ and a joint distribution $t'_\varepsilon$ such that
	\begin{align}
		(p \otimes r, q \otimes \eta) & \succ (t'_\varepsilon , q' \otimes \eta), \\
		 \frac{1}{2}\left\Vert p' - p'_\varepsilon \right\Vert_1 & \leq \varepsilon,\\
		 \D{}{t'_\varepsilon}{p'_\varepsilon \otimes r} & < \gamma,
\end{align}
where $t'_\varepsilon$ has first marginal $p'_\varepsilon$ and second marginal $r$,
	and  $\eta$ is the uniform distribution on the support of $r$.

The relative entropy $\D{}{t'_\varepsilon}{p'_\varepsilon \otimes r}$ can be made as small as desired and thus the output of the classical channel $\Lambda$ arbitrarily close to the product of $p'_\varepsilon$ and $r$. Furthermore, the state $p'_\varepsilon$ can be arbitrarily close to the target distribution $p'$ in normalized trace distance. Note that allowing for a slight error in the transformation of $p$ while demanding no error in the transformation of $q$ is consistent with the frameworks of \cite{matsumoto2010reverse,WW19,Buscemi2019}, with the framework of \cite{WW19} being known as the ``resource theory of asymmetric distinguishability'' due to this asymmetry in the transformation.\\

We now consider a special case of Theorem~\ref{Theorem2}, where we fix both $q$ and $q'$ to be the uniform distribution. A closely related, yet different claim is given as \cite[Lemma~3.10]{L16thesis}.

\begin{corollary}
\label{Corollary1}
Given probability distributions $p$ and $p'$ on the same alphabet such that $p \neq P p'$ for all permutation matrices $P$, the following conditions are equivalent:
\begin{enumerate}
	\item $H(p) \leq H(p')$
	\item For all $\varepsilon \in (0,1)$ and $\gamma>0$, there exist probability distributions $r$ and $p'_\varepsilon$, a joint distribution $t'_\varepsilon$, and a unital classical channel $\Lambda$ such that
	\begin{enumerate}
		\item $\Lambda(p \otimes r) = t'_\varepsilon$ with marginals $p'_\varepsilon$ and $r$,
		\item $\frac{1}{2} \left\Vert p' - p'_\varepsilon \right\Vert_1 \leq \varepsilon$,
		\item $\D{}{t'_\varepsilon}{p'_\varepsilon \otimes r} < \gamma$.
	\end{enumerate}
\end{enumerate}
\end{corollary}

A unital classical channel is one that preserves the uniform distribution. The corollary states that, given a catalyst, the feasibility condition for the transformation is based exclusively on the Shannon entropy. Furthermore, the relative entropy $\D{}{t'_\varepsilon}{p'_\varepsilon \otimes r}$ can be made as small as desired and thus the output of the classical channel $\Lambda$ arbitrarily close to the product of $p'_\varepsilon$ and $r$. Lastly, the state $p'_\varepsilon$ can be arbitrarily close to the target distribution $p'$ in normalized trace distance.

We note here that Corollary~\ref{Corollary1} applies more generally to quantum states $\rho$ and $\rho'$ of the same dimension and with  different spectra, where the first statement involves the von Neumann entropy and the second statement has $\Lambda$ become a quantum channel that incorporates a local unitary performing a change of basis from the eigenbasis of $\rho$ to the eigenbasis of $\rho'$. 

\section{Discussions and Implications}

\label{sec:Discuss}

In this section, we discuss various implications of our results. A catalyst is a special ancillary probability distribution whose presence enables certain transformations that would otherwise be impossible. As a consequence, broadening the class of transformations allowed in turn broadens the set of input and output distributions satisfying the mathematical condition governing the feasibility of a transformation \cite{Jonathan1999}. Additionally, allowing for the catalyst to be correlated with the target distribution, subject to the condition that its reduced state is returned unchanged, previously impossible transformations become possible.
Since the catalyst is returned unchanged, the paradigm of catalytic majorization with correlations is physically relevant and justified. In this latter case, further broadening the class of transformations allowed even simplifies the mathematical condition governing feasibility of a transformation to a single entropic condition. 

In the resource theories of pure-state entanglement, purity, and pure-state coherence, state transformations in the absence of a catalyst are governed by the majorization partial order, whereas in the resource theory of quasi-classical thermodynamics, transformations are governed by thermo-majorization \cite{HO13}. The catalytic majorization condition is equivalent to an infinite set of conditions on the R\'enyi entropy \cite{klimesh2007inequalities, Turgut_2007}. The  catalytic majorization condition is not as strict as the majorization condition, as there exist distributions $p$ and $q$ for which $p$ does not majorize $q$ but $p$ catalytically majorizes $q$ \cite{Jonathan1999}, establishing that catalysis provides an advantage.

As mentioned above, catalytic majorization can be further relaxed. In catalytic majorization, the target state, after transformation, is required to be in a product state with the catalyst; i.e., the catalyst and the target state are independent after the transformation. Relaxing this additional assumption and allowing the final state to be a non-product state, with the condition that the local states remain unchanged, further broadens the transformations allowed and simplifies the mathematical condition governing the feasibility of a transformation. This is a generalization of catalytic majorization, called catalytic majorization with correlations, as a product state is a special case of a general distribution. 

Our result in Theorem~\ref{Theorem2} is a necessary and sufficient condition for the last case, where we allow catalysis and correlations. The result is a single entropic condition involving relative entropy, and it neatly ties together catalytically correlated transformations in the resource theories of asymmetric distinguishability, quasi-classical thermodynamics, and purity into a single paradigm, due to the generality of our setting. Our result also introduces a different operational meaning for the relative entropy, by establishing its relevance in the single-shot regime. 


\begin{table*}
\begin{center}
    \begin{tabular}{|c||cl|cl|cl|Ml|}
        \hline 
         \diagbox[]{Resource\\Theory}{Regime} & \multicolumn{2}{c|}{Asymptotic}    & \multicolumn{2}{c|}{One-Shot} &
         \multicolumn{2}{c|}{One-shot + Catalysis}
           &  \multicolumn{2}{c|}{One-shot + Catalysis + Corr.}  \\
        \toprule 
        Entanglement & $\frac{H(p)}{H(p')}$  & \cite{BBPS96} & $p \prec p'$ & \cite{Nielsen1999} & $H_\alpha(p) \geq H_\alpha(p')$ &  \cite{klimesh2007inequalities,Turgut_2007} & ? 
         & \\
        \hline
        Coherence & $\frac{H(p)}{H(p')}$ & \cite{Winter_2016} & $p \prec p'$ & \cite{DBG15} & $H_\alpha(p) \geq H_\alpha(p')$ & \cite{BSW16} & ? & 
        \\
        \hline
        Thermodynamics & $\frac{\D{}{p}{\gamma}}{\D{}{p'}{\gamma}} $ & \cite{BHORS13}  & $p \succ_{\text{thermo}}  p'$ & \cite{HO13} & $F_\alpha(p) \geq F_\alpha(p')$ & \cite{Brand_o_2015} & $F(p) \geq F(p')$ & \cite{M_ller_2018} \\
        \hline
        Distinguishability & $\frac{\D{}{p}{q}}{\D{}{p'}{q'}}$&  \cite{WW19,Buscemi2019} & $(p, q) \succ (p', q')$ & \cite{Blackwell1953,Dahl99} & $\D{\alpha}{p}{q} \geq \D{\alpha}{p'}{q'}$ &  \cite{Brand_o_2015} & $\bm{\D{}{p}{q} \geq \D{}{p'}{q'}}$ & \\
        \hline
    \end{tabular}
\caption{Results in various regimes and resource theories: the table entries feature conditions for determining whether a transformation is possible in a given regime and resource theory. In all cases, the initial resource is unprimed and labeled by $p$ or $(p,q)$, and the target resource is primed and labeled by $p'$ or $(p',q')$. The bottom right cell with a bold entry indicates our contribution.  The one-shot regime with catalysis depends on $H_\alpha$ or $D_\alpha$ for all $ \alpha \in [-\infty, \infty]$. In the one-shot cases, by broadening the class of transformations allowed, more distributions can satisfy the mathematical conditions governing the feasibility of a transformation. When catalytic transformations with correlations are allowed, the mathematical condition governing feasibility of a transformation becomes a simple entropic condition. In the various resource theories listed in the table, $p$ refers to different objects: for entanglement theory, it denotes the Schmidt coefficients, for coherence, the incoherent coefficients, and for thermodynamics, the block-diagonal energy coefficients (see, e.g.,~\cite{PhysRevLett.122.110403}). We have placed question marks in the table entries that have not been explored (to the best of our knowledge).}
\label{table:results}
\end{center}
\end{table*}

\section{Related work}

There are some regimes of operating that are complementary to the ones given in our paper. This section provides a brief overview of these regimes and recalls some known results.

Classical relative majorization in \eqref{eq:rel-maj-notation} is known to be equivalent to the hypothesis testing region of $(p',q')$ being contained in the hypothesis testing region of $(p,q)$ \cite{Blackwell1953} (see also, e.g.,  \cite{Dahl99}, where the region is called a ``zonotope'', or \cite{Renes_2016} with motivations coming from quantum resource theories).
Quantum relative majorization was defined recently in \cite{Buscemi_2017} as an intended generalization of classical relative majorization concept in \eqref{eq:rel-maj-notation}, in which we write $(\rho,\sigma) \succ (\rho',\sigma')$, where $\rho$, $\sigma$, $\rho'$, and $\sigma'$ are quantum states, to indicate that the hypothesis testing region for $(\rho',\sigma')$ lies inside that of $(\rho,\sigma)$. Other equivalent conditions are given in \cite{Buscemi_2017}. In the case of qubit systems, the condition $(\rho,\sigma) \succ (\rho',\sigma')$ is equivalent to there existing a quantum channel $\mathcal{N}$ such that \cite{Alberti1980}
\begin{align}
\mathcal{N}(\rho) & = \rho', \label{eq:q-ch-transform-1} \\
\mathcal{N}(\sigma) & = \sigma' \label{eq:q-ch-transform-2},
\end{align}
which is consistent with \eqref{eq:rel-maj-notation}.
As stated above, for the fully classical case, which  has apparently motivated much of the related work in quantum information, the containment of testing regions is equivalent to  \eqref{eq:rel-maj-notation} \cite{Blackwell1953}. 
However, the equivalence no longer holds as soon as the dimension of the input states is three, with the dimension of the output states still being two \cite{matsumoto2014example}. Regardless, given states $\rho$, $\sigma$, $\rho'$, and $\sigma'$,
the existence of a quantum channel $\mathcal{N}$ satisfying \eqref{eq:q-ch-transform-1}--\eqref{eq:q-ch-transform-2} can be decided by means of a semi-definite program \cite{Gour_2018} (see also
\cite{WW19}). This latter paper \cite{Gour_2018} also developed a more general notion of majorization, called quantum majorization, which is based on \eqref{eq:q-ch-transform-1}--\eqref{eq:q-ch-transform-2}, but generalizes it further.

Majorization and relative majorization are most prominent in the single-shot regime, in which there is no assumed structure on the states involved. There exists another regime that deals with situations in which many independent and identically distributed (i.i.d.)~systems are available, often referred to as the asymptotic regime. In the asymptotic regime, perhaps the more pertinent question is that of rate of conversion from initial to target distributions and not necessarily feasibility of the transformation. For the resource theories of pure-state entanglement, pure-state coherence, and  quasi-classical thermodynamics, the rate of conversion is given by ratios of the entropy (pure-state entanglement and pure-state coherence) and relative entropy to the Gibbs distribution (quasi-classical thermodynamics). The asymptotic version of Blackwell's ``comparison of experiments'' paradigm was studied in \cite{matsumoto2010reverse, WW19, Buscemi2019}. In \cite{WW19}, the problem was approached by introducing the resource theory of asymmetric distinguishability, which allowed for analyzing this problem in a resource-theoretic way. The result of \cite{WW19, Buscemi2019} is that, in the asymptotic regime, the rate of conversion between two pairs of probability distributions $(p , q)$ and $(p' , q')$ is given by the ratio of the relative entropy of the pairs, thus enhancing the fundamental operational meaning of the relative entropy. This asymptotic result generalizes that from \cite{BHORS13}, which focused on the resource theory of thermodynamics in which the second states in the initial and final pairs are set to the thermal states.

Some prior results in various regimes and resource theories are presented in Table~\ref{table:results}.
As indicated in our paper and the table, our results are relevant for the single-shot regime in the context of relative catalytic majorization with correlations and are thus complementary to this prior work.

\section{Conclusion}

\label{sec:conclusion}

The majorization partial order was developed to capture the notion of disorder in probability distributions and governs transformations in various resource theories. The addition of a catalyst allows for transformations that are not possible in the framework of majorization without catalysis. Furthermore, allowing correlations between the target and the catalyst further broadens the class of transformations allowed and simplifies the feasibility condition for a transformation to a simple entropic condition.

The main result of our paper is that a pair $(p, q)$ of probability distributions can be converted approximately to another pair $(p', q')$ of probability distributions assisted by a catalyst with correlations if and only if $D(p \Vert q) \geq D(p' \Vert q')$. As mentioned earlier, we focus on the special case in which the second distribution $s$ of the catalyst pair is set to the uniform distribution. Our result thus extends the operational meaning of the relative entropy beyond the traditional i.i.d.~regime to the single-shot regime. This finding thus complements and enhances the previous finding from \cite{Boes_2019} in the context of entropy and the resource theory of purity.

There are several open questions yet to be tackled. The most pertinent is whether our result applies to pairs of general quantum states $(\rho, \sigma)$ and $(\tau, \omega)$ and quantum catalysts. This appears to be a very challenging question that will likely require new techniques. It would also be interesting to generalize the main results to continuous probability density functions, and it seems that the methods of \cite{RSS80} should be helpful in this regard. Another avenue to be explored is whether our result applies to larger sets of distributions, and not just pairs, i.e., the conversion of the set $\{p_1, p_2, \ldots, p_k\}$ of $k$ probability distributions to another set $\{q_1, q_2, \ldots, q_k\}$ of $k$ probability distributions for $k>2$. Lastly, another open question is whether the second element $\eta$ of the catalyst in Theorems~\ref{Theorem1} and \ref{Theorem2}  can feature an  arbitrary probability distribution instead of the uniform distribution. 

\begin{acknowledgments}
We thank Francesco Buscemi, Vishal Katariya, and Aby Philip for discussions. We are grateful to the anonymous referees for comments that helped to improve our paper. We acknowledge support from NSF Grant No.~1907615. SR thanks the Department of Physics and Astronomy at Louisiana State University for hospitality during his research visit in Fall 2019. SR also thanks the Department of Physics, Birla Institute of Technology and Science for this opportunity. MMW acknowledges support from Stanford QFARM and AFOSR (FA9550-19-1-0369).
\end{acknowledgments}

\bibliographystyle{unsrt}

\bibliography{draft2}

\appendix

\section{Notations and Definitions}

\label{sec:Notat}

We review various distinguishability measures and their properties in this appendix. Furthermore, we recall the definition of majorization and the embedding channel, which are used in the forthcoming theorems and lemmas.

\subsection{R\'enyi Divergence}

\label{subsec:RD}

For probability distributions $p = \{ p_i\}_{i=1}^k$ and $q = \{ q_i\}_{i=1}^k$, the \RD is defined for all $ \alpha \in \mathbb{R} \setminus \{ 0, 1, \infty, -\infty \}$ as follows \cite{renyi1961}:
\begin{equation}
	\D{\alpha}{p}{q} := \frac{\sgn{\alpha}}{\alpha-1} \ln \sum\limits_{i=1}^k p_i^\alpha q_i^{1-\alpha},
\end{equation}
where
\begin{equation}
        \sgn{\alpha} := 
        \begin{cases}
            1 & \alpha \geq 0 \\
            -1 & \alpha <0.
        \end{cases}
    \end{equation}
    The \RD can alternatively be written as follows:
    \begin{align}
        \D{\alpha}{p}{q} & = \frac{\sgn{\alpha}}{\alpha-1} \ln \sum\limits_{i=1}^k p_i \left(\frac{p_i}{q_i}\right)^{\alpha-1} \\
        & = \frac{\sgn{\alpha}}{\alpha-1} \ln \sum\limits_{i=1}^k q_i \left(\frac{p_i}{q_i}\right)^{\alpha},
    \end{align}
    which allows one to think of the expressions inside the logarithm as averages of tilted likelihood ratios.
Note that when $\alpha < 0$, the following identity holds
\begin{equation}
    \D{\alpha}{p}{q} = \frac{|\alpha|}{|\alpha|+1} \D{|\alpha|+1}{q}{p},
    \label{eq:neg-alpha-id-renyi-div}
\end{equation}
so that properties of $\D{\alpha}{p}{q}$ for negative values of $\alpha$ can be deduced from properties of $\D{\alpha}{p}{q}$ for values of $\alpha > 1$.

For $\alpha \in \{ 0, 1, \infty, -\infty\}$, we evaluate $\D{\alpha}{p}{q}$ as a limit to these values and arrive at the following expressions:
\begin{align}
\D{0}{p}{q} &= -\ln \sum\limits_{i=1:p_i\neq 0}^k q_i, \\
\D{1}{p}{q} &= \sum\limits_{i=1}^k p_i \ln \left(\frac{p_i}{q_i}\right),  \\
\D{\infty}{p}{q} &=\ln \max_i \frac{p_i}{q_i},  \\
\D{-\infty}{p}{q} &= \D{\infty}{q}{p}.
\end{align}
Note that $\D{\infty}{p}{q}$ can also be expressed as 
\begin{equation}
\D{\infty}{p}{q} =\ln \min \left\{ \lambda\ \middle| \ \lambda \geq \frac{p_i}{q_i} \ \forall i \right\}.
\end{equation}

An important property of the \RD is that it obeys the data-processing inequality \cite{van_Erven_2014}; i.e., for all classical channels $\Lambda$ and $ \alpha \in [-\infty, \infty]$, the following inequality holds
\begin{equation}
    \label{RenDivDatProc}
	\D{\alpha}{p}{q} \geq \D{\alpha}{\Lambda(p)}{\Lambda(q)}.
\end{equation}

Another important property of the \RD is that, for $\alpha \in [-\infty, \infty]$, it is non-negative
\begin{equation}
    \label{RenDivGreater0}
    \D{\alpha}{p}{q} \geq 0
\end{equation}
for all probability distributions $p$ and
$q$, and for all $\alpha \in [-\infty, \infty]\setminus \{0\}$, it is equal to zero
\begin{equation}
    \label{RenDivEqual0}
    \D{\alpha}{p}{q} = 0
\end{equation} 
if and only if $p = q$. For $\alpha = 0$, it is equal to zero if $p = q$, but the converse is not generally true.

The \RD is additive for product probability distributions \cite{van_Erven_2014}; i.e., for probability distributions $p_1$, $p_2$, $q_1$, and $q_2$, the following equality holds
\begin{equation}
    \label{RenDivAdd}
    \D{\alpha}{p_1 \otimes p_2}{q_1 \otimes q_2} = \D{\alpha}{p_1}{q_1} + \D{\alpha}{p_2}{q_2} .
\end{equation}

\subsection{R\'enyi Entropy}
\label{subsec:RE}

For a probability distribution $p = \{ p_i\}_{i=1}^k$, the \RE is defined for all $ \alpha \in \mathbb{R} \setminus \{ 0, 1, \infty, -\infty \}$ as follows~\cite{renyi1961}:
\begin{equation}
	H_\alpha(p) := \frac{\sgn{\alpha}}{1-\alpha}\ln \sum\limits_{i=1}^k p_i^\alpha.
\end{equation}
Similarly to the R\'enyi divergence, we also define the \RE for $\alpha \in \{ 0, 1, \infty, -\infty\}$ as a limit to these values, and the result is the following expressions:
\begin{align}
H_0(p) &=\ln \rank{p},  \\
H_1(p) &= -\sum\limits_{i=1}^k p_i\ln p_i,  \\
H_\infty(p) &= -\ln p_{\max},  \\
H_{-\infty}(p) &= \ln p_{\min}.
\end{align}
The following equality relates the \RD to the \RE for all $\alpha \in [-\infty,\infty]$:
\begin{equation}
     D_{\alpha}(p \Vert \eta_k) = \ln k - H_\alpha(p),
     \label{eq:renyi-div-to-unif-relate-ent}
\end{equation}
where $\eta_k$ is the uniform distribution on the alphabet of~$p$.

\subsection{Quantum R\'enyi Divergence}
\label{subsec:QRD}

The \RD from the classical regime can be extended to the quantum regime, but the extension is not unique. Classical probability distributions can be modeled as quantum density operators that are diagonal in a fixed orthonormal basis. Thus, modelling two probability distributions is equivalent to working with two density operators that commute with each other, i.e., that are diagonal in the same basis. However, quantum states do not generally commute, which leads to the non-uniqueness of the quantum extension. Note that any quantum extension should collapse to the classical case if the states commute.

Here, we employ the following quantum \RD   \cite{PETZ198657}:
\begin{equation}
	D_\alpha(\rho \Vert \sigma) := \frac{\sgn{\alpha}}{\alpha-1}\ln \Tr( \rho^\alpha \sigma^{1-\alpha}).
\end{equation}
The \RD $D_\alpha$ converges to the following expressions in the limit as $\alpha$ tends to either zero or one:
\begin{align}
D_0(\rho \Vert \sigma) &= -\ln \Tr(\Pi_\rho \sigma) ,  \\
D_1(\rho \Vert \sigma) &= \Tr(\rho(\ln \rho - \ln \sigma)),
\end{align}
where $\Pi_\rho$ is the projector onto the support of  $\rho$.  The quantum \RD obeys the data-processing inequality for all quantum channels $\Lambda$ and $\alpha \in [0 , 2]$ \cite{PETZ198657}:
\begin{equation}
D_\alpha(\rho \Vert \sigma) \geq D_\alpha(\Lambda(\rho) \Vert \Lambda(\sigma)).
\end{equation}
Another property of the quantum \RD is that for states $\rho$ and $ \sigma$ \cite{PETZ198657}, it is non-negative
\begin{equation}
\label{QuaRenDiv>0}
    \SD{\alpha}{\rho}{\sigma} \geq 0,
\end{equation}
and for all $\alpha \in [-\infty,\infty]\setminus \{0\}$, it is equal to zero
\begin{equation}
    \SD{\alpha}{\rho}{\sigma} = 0
\end{equation}
if and only if $\rho = \sigma$. If $\alpha = 0$, it is equal to zero if $\rho = \sigma$ but the converse is not generally true.

\subsection{Majorization and its Extensions}
\label{subsec:mjrz}

The theory of majorization was developed to capture the notion of disorder. The crux of the theory is to answer the question: When is one probability distribution more disordered than another? The results can be succinctly put forth as follows \cite{G.H.Hardy1952}.

For two probability distributions $p, q \in \mathbb{R}^k $, we say that ``$p$ majorizes $q$'' or $p \succ q$ if for all $n \in \{ 1, \ldots, k-1\}$,
\begin{align}
	\sum\limits_{i=1}^n p_i^\downarrow &\geq \sum\limits_{i=1}^n q_i^\downarrow,  \\
	\sum\limits_{i=1}^k p_i^\downarrow &= \sum\limits_{i=1}^k q_i^\downarrow,
\end{align}
where $p^\downarrow$ is the reordering of $p$ in descending order. If $p \succ q$, it follows that $q$ is more disordered than $p$ in a precise sense stated in Lemma \ref{lem:mjrz} in Appendix~\ref{sec:TechLem}. That is, there are several formulations of majorization that are known to be equivalent.
 
An extension of the theory of majorization comes in the form of catalysis. Given two distributions $p$ and $q$ such that $p \not \succ q$, does there exist a distribution $r$ such that $p \otimes r \succ q \otimes r$? If such a distribution exists, it can be used to transform $p$ into $q$ and be left unchanged. Since $r$ is unaffected by the process, we call it the catalyst of the transformation. We say that $q$ is catalytically majorized by $p$, written $p \succ_T q$, if there exists a catalyst $r$ such that $p \otimes r \succ q \otimes r$. The theory of catalytic majorization is not as well understood as majorization. However, several important results have been established \cite{klimesh2007inequalities, Turgut_2007}.

\subsection{Embedding Channel}

\label{subsec:embed}

A mathematical tool frequently used in the resource theory of thermodynamics is an embedding channel, which is a classical channel that maps a thermal distribution to a uniform distribution \cite{Brand_o_2015}. We use this tool in a more general sense needed in our context here, in order to map a probability distribution described by a set of rational numbers to a uniform distribution. Note that the embedding channel was originally introduced in \cite[page 227]{RSS80} for continuous probability density functions.

Consider the simplex $\{p_i\}_{i=1}^k$ of probability distributions  and a set $d:=\{d_i\}_{i=1}^k$ of natural numbers  with
\begin{equation}
N := \sum\limits_{i=1}^k d_i.    
\end{equation}
The embedding channel $\Gamma_d : \mathbb{R}^k \rightarrow \mathbb{R}^N$, corresponding to the set $d$, is defined as
\begin{align}
\Gamma_d(p) & := \bigoplus_{i=1}^k p_i \eta_i \\
& = \left( \underbrace{\frac{p_1}{d_1}, \ldots, \frac{p_1}{d_1}}_{d_1}, \ldots, \underbrace{\frac{p_k}{d_k}, \ldots, \frac{p_k}{d_k}}_{d_k} \right),
\end{align}
where $\eta_i = \left( \frac{1}{d_i}, \ldots, \frac{1}{d_i} \right) \in \mathbb{R}^{d_i}$ is the uniform distribution.

There exists a left inverse for the embedding channel~$\Gamma_d$, which is denoted by $\Gamma^*_d : \mathbb{R}^N \rightarrow \mathbb{R}^k$ and defined as follows:
\begin{equation}
\label{PsInvEmbed}
\Gamma^*_d(x) := \left( \sum\limits_{i=1}^{d_1} x_i, \sum\limits_{i=d_1+1}^{d_1 + d_2} x_i, \ldots, \sum\limits_{i = d_1 + \ldots + d_{n-1}+1}^{N} x_i\right).
\end{equation}
Note that $\Gamma^*_d$ is itself a classical channel, and furthermore, it reverses the action of $\Gamma_d$ in the following sense:
\begin{equation}
\Gamma^*_d \circ \Gamma_d = \operatorname{id},   
\label{eq:left-inv-ch-embed}
\end{equation}
where $\operatorname{id}$ is the identity channel.

\section{Technical Lemmas}

\label{sec:TechLem}

Several of the technical lemmas presented in this appendix have been established in prior work. Here we list them, with modifications and extensions, for convenience and completeness. 

\begin{lemma}[Majorization \cite{marshall2011inequalities, G.H.Hardy1952}]
\label{lem:mjrz}
Let $x,y\in  \mathbb{R}^k$ be probability distributions. Then the following are equivalent:
\begin{enumerate}
	\item $x \succ y$.
	
	\item For $n \in \{ 1, \ldots, k-1\}$, \begin{align}
	    \sum\limits_{i=1}^n x_i^\downarrow & \geq \sum\limits_{i=1}^n y_i^\downarrow , \\ \sum\limits_{i=1}^k x_i^\downarrow & = \sum\limits_{i=1}^k y_i^\downarrow .
	\end{align}
	
	\item $y = D x$ for some $k \times k$ doubly stochastic matrix $D$.
	
	\item For all $t \in \mathbb{R}$, $\sum\limits_{i=1}^k \vert x_i - t \vert \geq \sum\limits_{i=1}^k \vert y_i - t \vert$.
\end{enumerate}
\end{lemma}

\begin{lemma}[Embedding of distribution with rational entries \cite{Brand_o_2015}]
\label{lem:embedRational}
Let  $\gamma_d$ be a probability distribution defined as follows:
\begin{equation}
\gamma_d := \left( \frac{d_1}{N},\ldots, \frac{d_k}{N} \right),
\end{equation}
where $d:= \{d_i\}_{i=1}^k$ is a set of natural numbers such that
\begin{equation}
N := \sum\limits_{i=1}^k d_i .    
\end{equation}
 Then the action of the embedding channel $\Gamma_d$ on $\gamma_d$ is as follows:
\begin{equation}
\Gamma_d(\gamma_d) = \eta_N,
\end{equation}
where $\eta_N$ is the uniform distribution of size $N$.
\end{lemma}
\begin{proof}
From the definition of an embedding channel $\Gamma_d$ corresponding to the same set $d$ (see Subsection \ref{subsec:embed}), we find that
\begin{align}
\Gamma_d(\gamma_d) &= \left( \underbrace{\frac{d_1/N}{d_1}, \ldots, \frac{d_1/N}{d_1}}_{d_1}, \ldots, \underbrace{\frac{d_k/N}{d_k}, \ldots, \frac{d_k/N}{d_k}}_{d_k} \right), \nonumber \\
&= \left( \underbrace{1/N, \ldots, 1/N}_{d_1}, \ldots, \underbrace{1/N, \ldots, 1/N}_{d_k} \right), \nonumber \\
&= \left( \underbrace{1/N, \ldots, 1/N}_{N}\right), \nonumber \\
&= \eta_N.
\end{align}
This concludes the proof.
\end{proof}

\begin{lemma}[Preservation of \RD \cite{Brand_o_2015}]
\label{lem:PreserveRenyi}
Let $p := \{p_i\}_{i=1}^k$ be a probability distribution and let $d:= \{d_i\}_{i=1}^k$ be a set of natural numbers with $N:=\sum\limits_{i=1}^k d_i$. Let $\tilde{p}$ denote the distribution $\tilde{p} := \Gamma_d(p)$. Then for all $\alpha \in [-\infty, \infty]$, the following equality holds
\begin{equation}
\D{\alpha}{p}{\gamma_d} = \D{\alpha}{\tilde{p}}{\eta_N},
\end{equation}
where $\gamma_d$ is as defined in Lemma \ref{lem:embedRational} and $\Gamma_d$ is as defined in Subsection \ref{subsec:embed}.
\end{lemma}

\begin{proof}
This is a consequence of the data-processing inequality in \eqref{RenDivDatProc} holding for all $\alpha \in [-\infty, \infty]$, as well as the identity in \eqref{eq:left-inv-ch-embed}. Indeed, for $\alpha \in [-\infty, \infty]$, the following inequality holds as a consequence of data processing:
\begin{equation}
    \D{\alpha}{p}{\gamma_d} \geq \D{\alpha}{\Gamma_d(p)}{\Gamma_d(\gamma_d)} = \D{\alpha}{\tilde{p}}{\eta_N}.
\end{equation}
Additionally, the following inequality is a consequence of data processing and the identity in \eqref{eq:left-inv-ch-embed}:
\begin{align}
    \D{\alpha}{\tilde{p}}{\eta_N}
    & = \D{\alpha}{\Gamma_d(p)}{\Gamma_d(\gamma_d)}\\
    & \geq \D{\alpha}{(\Gamma_d^* \circ \Gamma_d)(p)}{(\Gamma_d^* \circ \Gamma_d)(\gamma_d)} \\
    & =
    \D{\alpha}{p}{\gamma_d} .
\end{align}
This concludes the proof.
\end{proof}

\begin{lemma}[Mixture reduces \RD \cite{Brand_o_2015}]
\label{lem:MixRenyi}
Let $p$ and $q$ be  probability distributions such that $p \neq q$ and  $q$ is full rank. Then $\forall \alpha \in (-\infty, \infty) \setminus \{0\}$ and $0<\delta<1$,
\begin{equation}
\D{\alpha}{(1-\delta)p + \delta q}{q} < \D{\alpha}{p}{q}.
\end{equation}
\end{lemma}

\begin{proof}
We split the proof into separate parts for different $\alpha$ values.

For $\alpha \in (0, 1]$, the following inequality is a consequence of joint convexity of $D_\alpha$ for probability distributions $(A_1, B_1)$ and $(A_2, B_2)$ and a parameter $0<\delta<1$, \cite{van_Erven_2014}:
\begin{multline}
     \D{\alpha}{(1-\delta)A_1 + \delta A_2}{(1-\delta)B_1 + \delta B_2}  \\
	  \leq (1-\delta)\D{\alpha}{A_1}{B_1} + \delta \D{\alpha}{A_2}{B_2}.
\end{multline}
Setting $B_1 = B_2 = A_2 = q$ and $A_1 = p$, we find that
\begin{align}
    \D{\alpha}{(1-\delta)p + \delta q}{q} &\leq (1-\delta)\D{\alpha}{p}{q} \nonumber \\
	&< \D{\alpha}{p}{q},
\end{align}
where the first inequality follows because $\D{\alpha}{q}{q} = 0$ (see \eqref{RenDivEqual0}) and the second strict inequality follows because $\D{\alpha}{p}{q} > 0$ for $p \neq q$ (see \eqref{RenDivGreater0}).

For $\alpha > 1$, joint convexity does not hold, and so a different approach is required. Let us define $r := \left(1-\delta\right)p + \delta q$. Then the statement of the lemma is equivalent to the following:
\begin{equation}
    \label{Mixture-equivalent}
	\sum\limits_{i=1}^k r_i^\alpha q_i^{1-\alpha} < \sum\limits_{i=1}^k p_i^\alpha q_i^{1-\alpha}.
\end{equation}
For $\alpha>1$, $f(x) = x^\alpha$ is a convex function over all $x>0$. Since $q_i^{1-\alpha}>0$, the following function
\begin{equation}
	Q_{\alpha}(p) := \sum\limits_{i=1}^k p_i^\alpha q_i^{1-\alpha}
\end{equation}
is convex with respect to $p$. For $p \neq q$, $Q_{\alpha}(p)>1$, which follows from strict positivity of the \RD $\D{\alpha}{p}{q}$. Note that $Q_{\alpha}(q) = \sum\limits_{i=1}^k q_i = 1$. Thus,
\begin{align}
	Q_{\alpha}(r) &\leq (1-\delta)Q_{\alpha}(p) + \delta Q_{\alpha}(q), \nonumber \\
	&= Q_{\alpha}(p) - \delta(Q_{\alpha}(p) - Q_{\alpha}(q)), \nonumber \\
	&< Q_{\alpha}(p),
\end{align}
where the first inequality follows from convexity and the second inequality follows because $Q_{\alpha}(p) > Q_{\alpha}(q)$. Thus, we have established \eqref{Mixture-equivalent}.
	
For $\alpha < 0$, if $p$ is not full rank, then $\D{\alpha}{p}{q} = \infty$. However, $q$ is full rank, implying that $\D{\alpha}{(1-\delta)p + \delta q}{q}$ is finite. Thus, the desired inequality holds. If $p$ is full rank, we use a similar approach as that used in the case  $\alpha>1$, by appealing to the identity in \eqref{eq:neg-alpha-id-renyi-div}, thus concluding the proof.
\end{proof}

\begin{lemma}
\label{lem:rewrite-TD}
Let $r$ and $s$ be probability distributions over the same alphabet. Then
\begin{equation}
    \frac{1}{2} \left\Vert r - s \right\Vert_1 = \sum_{i : r_i > s_i} r_i - s_i .
\end{equation}
\end{lemma}

\begin{proof}
Consider that
\begin{equation}
    1 = \sum_{i : r_i > s_i} r_i + \sum_{i : r_i \leq s_i} r_i = 
    \sum_{i : r_i > s_i} s_i + \sum_{i : r_i \leq s_i} s_i.
\end{equation}
This implies that
\begin{equation}
    \sum_{i : r_i > s_i} r_i - s_i = 
    \sum_{i : r_i \leq s_i} s_i - r_i.
\end{equation}
Then we find that
\begin{align}
     \left\Vert r - s \right\Vert_1  & =  \sum_i | r_i - s_i | \\
    & =  \left[\sum_{i : r_i > s_i} r_i - s_i \right] + \left[\sum_{i : r_i \leq s_i} s_i - r_i  \right] \\
    & = 2 \sum_{i : r_i > s_i} r_i - s_i.
\end{align}
This concludes the proof.
\end{proof}

\begin{definition}[Reversal channel] \label{rmrk:Reversal} As is well known, a classical channel $E$ is defined in terms of a conditional probability distribution. This implies that, given an input distribution $p(x)$ and a channel $E(y|x)$, the output distribution $p'(y)$ is given by
\begin{equation}
    \label{channelDef}
    p'(y) = \sum_x E(y|x) p(x).
\end{equation}
A reversal channel for  the classical channel $E$ on $p$ can be defined as follows. From the Bayes theorem, we know that
\begin{equation}
    \label{Bayes}
    E(y|x) p(x) = E'(x|y) p'(y).
\end{equation}
Summing both sides of  \eqref{Bayes} over $x$ and noticing that $\sum_x E'(x|y) = 1$ is the normalization condition, we recover \eqref{channelDef}.
Similarly, summing both sides of  \eqref{Bayes} over $y$, we see that
\begin{equation}
    p(x) = \sum_y E'(x|y) p'(y),
\end{equation}
and thus, a reversal channel $E'$ can be defined as
\begin{equation}
    E'(x|y) = \frac{E(y|x) p(x)}{p'(y)} .
\end{equation}
\end{definition}

Eqs.~\eqref{eq:q-rat-close}, \eqref{eq:perfect-for-q}, and \eqref{eq:approx-for-p} of the following lemma were established in \cite{Brand_o_2015}, and the proof given below follows the proof given there closely.

\begin{lemma}[Conversion into rational entries] 
\label{lem:ConvertRational}
Let $q = \{q_i\}_{i=1}^k$ be an ordered (descending) probability distribution of full rank. Then, for all $\varepsilon>0$, there exists a probability distribution $\tilde{q}$ such that 
	 \begin{equation}
	     \frac{1}{2} \left \Vert q - \tilde{q} \right\Vert_1 \leq \varepsilon,
	     \label{eq:q-rat-close}
	 \end{equation}
where $\tilde{q} = \{d_i/N\}_i$, for a set $\{d_i\}_{i=1}^k$ of natural numbers, with $N := \sum\limits_{i=1}^k d_i$. Additionally, there exists a classical channel $E$ such that
\begin{equation}
\tilde{q} = E(q),  
\label{eq:perfect-for-q}
\end{equation}
 and for all other probability distributions $p$, the following inequalities hold
 \begin{align}
     \frac{1}{2} \left\Vert p - E(p) \right\Vert_1
     & \leq O(\sqrt{\varepsilon}),
     \label{eq:approx-for-p} \\
     \frac{1}{2} \left\Vert p - R(p) \right\Vert_1
     & \leq O(\sqrt{\varepsilon}),
     \label{eq:approx-for-p-reversal-ch}
 \end{align}
 where $R$ is the reversal channel for $E$ on $q$, as given in Definition~\ref{rmrk:Reversal}.
\end{lemma}

\begin{proof}
To begin with, if $q$ is already rational, then the lemma trivially holds with $\tilde{q} = q$ and $E$ and $R$ set to the identity channel. 

So we move on to the non-trivial case in which $q$ is not rational. We first construct a probability distribution $\tilde{q}$ that satisfies \eqref{eq:q-rat-close}. Then we construct a classical channel $E$ that satisfies \eqref{eq:perfect-for-q} and \eqref{eq:approx-for-p}. We finally construct the reversal channel $R$ that satisfies \eqref{eq:approx-for-p-reversal-ch}.

Fix $\varepsilon > 0 $. Since $q$ is decreasing, it follows that $\min_i q_i = q_k$. 
 Pick
\begin{equation}
    N \geq \max\left\{\left(\frac{k+1}{q_k}\right)^2, \frac{k}{\varepsilon}, 4\right\}.
    \label{eq:choice-of-N-suff-large}
\end{equation}
Then $N$ is sufficiently large such that the following inequality holds
\begin{equation}
\label{L5-NChoice}
	q_k > \frac{k}{\sqrt{N}} > \frac{k}{N}.
\end{equation}

We now define $\tilde{q}$. For $i \in \{ 1, \ldots, k-1\}$,
set 
\begin{align}
    d_i & := \ceil{q_i N}, \\
\label{L5-qtildei}
	\tilde{q}_i & := \frac{d_i}{N},
\end{align}
and
\begin{align}
d_k & := N - \sum_{i=1}^{k-1} d_i, \\
\label{L5-qtildek}
    \tilde{q}_k & = \frac{d_k}{N} = 1- \sum\limits_{i=1}^{k-1} \tilde{q}_i  .
\end{align}
Note that $\tilde{q}_i \geq q_i$ for $i \in \{ 1, \ldots, k-1\}$  and $\tilde{q}_k \leq q_k$ because both distributions are normalized. In fact, we can conclude that
$\tilde{q}_k < q_k$ from the assumption that $q$ is not rational. Also note that, 
\begin{align}
1-\tilde{q}_k &= \sum\limits_{i=1}^{k-1} \tilde{q}_i  \nonumber \\
&= \sum\limits_{i=1}^{k-1} \frac{\ceil{q_i N}}{N}  \nonumber \\
&\leq  \sum\limits_{i=1}^{k-1} \frac{q_i N + 1}{N}  \nonumber \\
&= \left(\sum\limits_{i=1}^{k-1} \frac{q_i N}{N}\right) + \frac{k-1}{N} \nonumber \\
&\leq 1-q_k + \frac{k}{N},
\end{align}
where the first equality follows from \eqref{L5-qtildek} and the second equality follows from \eqref{L5-qtildei}.
Thus,
\begin{align}
\label{L5-qtildek>0}
	\tilde{q}_k \geq q_k - \frac{k}{N} &> 0,
\end{align}
which follows from \eqref{L5-NChoice} by the choice of $N$. Thus $\tilde{q}$ is a legitimate probability distribution of full rank. 
Now applying Lemma~\ref{lem:rewrite-TD}, we find that
\begin{align}
	\frac{1}{2} \left\Vert q - \tilde{q} \right\Vert_1  
	&= \sum\limits_{i : q_i > \tilde{q}_i} q_i - \tilde{q}_i, \nonumber \\
	&= q_k - \tilde{q}_k, \nonumber \\
	&\leq \frac{k}{N} \leq \varepsilon,
\label{q-tildeq}
\end{align}
where the last inequalities are due to \eqref{L5-qtildek>0} and \eqref{eq:choice-of-N-suff-large}.
Thus, it follows that \eqref{eq:q-rat-close} is satisfied.  

Now we construct a classical channel that takes $q$ to~$\tilde{q}$. Such a channel must slightly increase the probabilities $q_i$ for all  $i \in \{1, \ldots, k-1\}$, while reducing $q_k$. Recall that classical channels are characterised by conditional probabilities $P(j | i)$ satisfying
\begin{equation}
P(j | i) \geq 0\quad \forall i,j,\qquad 	\sum\limits_{j=1}^k P(j | i) = 1. \label{eq:cl-ch-conds}
\end{equation}

Set $\mathcal{I} := \{1, 2, \ldots, k-1 \}$, which is the set of  indices for which $\tilde{q}_i \geq q_i$. For $i \in \mathcal{I}$, let
\begin{align}
    \Delta_i & := \tilde{q}_i - q_i,\\
    \Delta & := \sum_{i \in \mathcal{I}} \Delta_i.
    \label{eq:delta-def-pf}
\end{align}
Note that $\tilde{q}_k = q_k - \Delta$. 

Consider a classical channel $E$ defined as follows:
\begin{equation}
	P(j | i) := 
	\begin{cases} 
	   \delta_{ij} &   i,j  \in \mathcal{I}\\
      	0 & i \in \mathcal{I}, j=k\\
      	\frac{\Delta_j}{q_k} & i=k, j \in \mathcal{I}\\
      	1- \frac{\Delta}{q_k} & i=j=k
   \end{cases}.
   \label{eq:cl-ch-rat-to-arb}
\end{equation}
It is clear that the non-negativity condition in \eqref{eq:cl-ch-conds} holds, and we now prove that 1) the normalization condition holds and 2) $E(q) = \tilde{q}$. We begin with the normalization condition.
For $i \in \mathcal{I}$,
we have that
		\begin{align}
			 \sum\limits_{j=1}^k P(j | i) 
			&= P(k | i) + P(i | i) + \sum\limits_{j \in \mathcal{I}, j\neq i} P(j | i)  , \nonumber \\
			&= 0 + 1 + 0, \nonumber \\
			&= 1.
		\end{align}
	 For $i = k$, we have that
		\begin{align}
			 \sum\limits_{j=1}^k P(j | k) 
			&= P(k | k) + \sum\limits_{j \in \mathcal{I}} P(j | k) , \nonumber \\
			&= 1- \frac{\Delta}{q_k} + \sum\limits_{j \in \mathcal{I}}  \frac{\Delta_j}{q_k}  ,  \nonumber \\
			&= 1.
		\end{align}
Thus, the following equality holds for all $i\in \{1, \ldots, k\}$:
\begin{equation}
	\sum\limits_{j=1}^k P(j | i) = 1,
\end{equation}
which corresponds to the preservation of normalization.
		
Next, we show that $E(q) = \tilde{q}$.  For $j \in \mathcal{I}$, we have that
		\begin{align}
			[E(q)]_j &= \sum\limits_{i =1}^k P(j | i)q_i \nonumber \\
			& =  P(j | j)q_j + P(j | k)q_k + \sum\limits_{i \in \mathcal{I}, i\neq j} P(j | i)q_i , \nonumber \\
			&= q_j + \frac{\Delta_j}{q_k}q_k + 0 , \nonumber \\
			&=  q_j + \Delta_j , \nonumber \\
			&= \tilde{q}_j.
		\end{align}
	For $j = k$, we have that
		\begin{align}
			[E(q)]_k & =  \sum\limits_{i =1}^k P(k | i)q_i \nonumber \\
			&= P(k | k)q_k + \sum\limits_{i \in \mathcal{I}} P(k | i)q_i , \nonumber \\
			&= \left(1-\frac{\Delta}{q_k} \right)q_k + 0 , \nonumber \\
			&= q_k - \Delta, \nonumber \\
			&= \tilde{q}_k.
		\end{align}
Thus, we have proven that $\tilde{q} = E(q)$.

We now need to show that, for all other distributions~$p$, the following inequality holds
\begin{equation}
\left \Vert p - E(p) \right\Vert_1 \leq O(\sqrt{\varepsilon}).    
\end{equation}
Let $\tilde{p} := E(p) $. If $\tilde{p} = p$, then the desired inequality trivially holds. Otherwise, consider that $\tilde{p}_j \geq p_j$ for all $j \neq k$ and $p_k > \tilde{p}_k $, due to the form of the classical channel $E$ in \eqref{eq:cl-ch-rat-to-arb}. Indeed, this follows because, for $j\neq k$, we have that
\begin{align}
    \tilde{p}_j
    & = \sum_{i=1}^k P(j|i) p_i \nonumber \\
    & = P(j|j) p_j + P(j|k)p_k + \sum_{i\in \mathcal{I} : i\neq j} P(j|i) p_i \nonumber \\
    & = p_j + \Delta_j p_k / q_k \nonumber  \\
    & \geq p_j,
\end{align}
while for $j = k$, we have that
\begin{align}
    \tilde{p}_k
    & = \sum_{i=1}^k P(k|i) p_i \nonumber \\
    & = P(k|k) p_k +  \sum_{i\in \mathcal{I} } P(k|i) p_i \nonumber \\
    & =   \left(1- \frac{\Delta} { q_k}\right) p_k  \nonumber \\
    & < p_k.
\end{align}
Now consider that
\begin{align}
	\frac{1}{2} \left \Vert p - \tilde{p} \right\Vert_1 &= p_k - \tilde{p}_k, \nonumber \\
	&= p_k - \left(1- \frac{\Delta}{q_k}\right)p_k, \nonumber \\
	&= \left( \frac{\Delta}{q_k}\right) p_k, \nonumber \\
	&\leq \frac{\Delta}{q_k}, \nonumber \\
	&\leq \frac{k}{N q_k}, \nonumber \\
	&\leq \frac{1}{\sqrt{N}}, \nonumber \\
	&\leq \sqrt{\frac{\varepsilon}{k}},
\end{align}
where the first equality follows from Lemma~\ref{lem:rewrite-TD} and the reasoning above, the first inequality follows because $p_k \leq 1$, the second inequality follows because $\Delta = q_k - \tilde{q}_k \leq \frac{k}{N}$, the third inequality follows because $ q_k \geq \frac{k}{\sqrt{N}}$, and the last inequality follows because $\sqrt{\frac{k}{N}} \leq  \sqrt{\varepsilon}$.
Thus, for all other distributions $p$, the following inequality holds $\left\Vert p - E(p) \right\Vert_1 \leq O(\sqrt{\varepsilon})$.

By applying Definition~\ref{rmrk:Reversal}, the reversal channel $R$ for $E$ on $q$ is defined as follows:%
\begin{equation}
R(i|j):=\frac{P(j|i)q_{i}}{\tilde{q}_{j}}.
\end{equation}
For $i,j\in\mathcal{I}$, this becomes%
\begin{equation}
R(i|j)=\frac{\delta_{i,j}q_{i}}{\tilde{q}_{j}}.
\end{equation}
For $i\in\mathcal{I}$ and $j=k$, we have%
\begin{equation}
R(i|j)=0.
\end{equation}
For $i=k$ and $j\in\mathcal{I}$, we have%
\begin{equation}
R(i|j)=\frac{\Delta_{j}}{q_{k}}\frac{q_{k}}{\tilde{q}_{j}}=\frac{\Delta_{j}%
}{\tilde{q}_{j}}.
\end{equation}
For $i=j=k$, we have%
\begin{equation}
R(i|j)=\left(  1-\frac{\Delta}{q_{k}}\right)  \frac{q_{k}}{\tilde{q}_{k}%
}=\frac{q_{k}}{\tilde{q}_{k}}-\frac{\Delta}{q_{k}}\frac{q_{k}}{\tilde{q}_{k}%
}=\frac{q_{k}-\Delta}{\tilde{q}_{k}}=1.
\end{equation}
Summarizing all of this, the reversal channel $R(i|j)$ can be written as%
\begin{equation}
R(i|j)=\left\{
\begin{array}
[c]{cc}%
\frac{\delta_{i,j}q_{i}}{\tilde{q}_{i}} & i,j\in\mathcal{I}\\
0 & i\in\mathcal{I},\ j=k\\
\frac{\Delta_{j}}{\tilde{q}_{j}} & i=k,\ j\in\mathcal{I}\\
1 & i=j=k
\end{array}
\right.  .
\end{equation}
The action of this reversal channel on an arbitrary probability distribution
$p$, leading to an output distribution $Rp$, is as follows. For $i\neq k$, we
have that%
\begin{align}
(Rp)_{i} &  =\sum_{j=1}^{k}R(i|j)p_{j} \nonumber\\
&  =R(i|k)p_{k}+R(i|i)p_{i}+\sum_{j\in\mathcal{I},j\neq i}R(i|j)p_{j}\nonumber\\
&  =0+\frac{q_{i}p_{i}}{\tilde{q}_{i}}+0\nonumber\\
&  =\frac{q_{i}p_{i}}{\tilde{q}_{i}}\nonumber\\
&  \leq p_{i},
\end{align}
and for $i=k$, we have that%
\begin{align}
(Rp)_{k} &  =\sum_{j=1}^{k}R(k|j)p_{j}\nonumber\\
&  =R(k|k)p_{k}+\sum_{j\in\mathcal{I}}R(k|j)p_{j}\nonumber\\
&  =p_{k}+\sum_{j\in\mathcal{I}}\frac{\Delta_{j}}{\tilde{q}_{j}}p_{j}\nonumber\\
&  >p_{k}.\label{eq:reversal-on-k}%
\end{align}
Then we establish the following bound:
\begin{align}
\frac{1}{2}\left\Vert Rp-p\right\Vert _{1} &  =(Rp)_{k}-p_{k}\nonumber\\
&  =\sum_{j\in\mathcal{I}}\frac{\Delta_{j}}{\tilde{q}_{j}}p_{j}\nonumber\\
&  \leq\frac{1}{\tilde{q}_{k}}\sum_{j\in\mathcal{I}}\Delta_{j}\nonumber\\
&  =\frac{\Delta}{\tilde{q}_{k}}\nonumber\\
&  \leq\frac{k}{N\tilde{q}_{k}}\nonumber\\
&  \leq\frac{2}{\sqrt{N}}\nonumber\\
&  \leq2\sqrt{\frac{\varepsilon}{k}}.
\end{align}
The first equality follows because $(Rp)_{k}>p_{k}$ and $(Rp)_{i}\leq p_{i}$
for $i\neq k$, and by applying Lemma~\ref{lem:rewrite-TD}. The second equality follows from
\eqref{eq:reversal-on-k}. The first inequality follows because $p_{j}\leq1$
and $\tilde{q}_{j}\geq\tilde{q}_{k}$ for all $j\in\mathcal{I}$. The third
equality follows from \eqref{eq:delta-def-pf}. The second inequality follows because $\Delta\leq
k/N$. The third inequality follows from applying \eqref{L5-qtildek>0} and \eqref{L5-NChoice} to conclude that
\begin{equation}
\tilde{q}_{k}\geq q_{k}-\frac{k}{N}\geq\frac{k}{\sqrt{N}}-\frac{k}{N} \geq 
\frac{k}{2\sqrt{N}}  , \label{eq:k-N-ineq-pf}
\end{equation}
with the last inequality in \eqref{eq:k-N-ineq-pf} following because
\begin{equation}
    \frac{1}{\sqrt{N}}-\frac{1}{N} = 
\frac{\sqrt{N}-1}{N}\geq\frac{\tfrac{1}{2} \sqrt{N}}{N} = \frac{1}{2\sqrt{N}}
\end{equation} when $N\geq4$, which is a consequence of \eqref{eq:choice-of-N-suff-large}.
  The final inequality follows
because $\sqrt{\frac{k}{N}}\leq\sqrt{\varepsilon}$. 
\end{proof}

\begin{remark} \label{rmrk:FullRank}
For a full rank distribution $p$, the distribution $E(p)$ as defined above is full rank. The reasoning is as follows: 
\begin{itemize}
	\item The channel increases all entries except the last entry. Thus, all entries except the last entry in $E(p)$ are strictly greater than zero because $p$ is full rank.
	\item The last entry in $E(p)$ is $\left( 1- \frac{\Delta}{q_k} \right) p_k$. Since $\Delta = q_k - \tilde{q}_k$, by the choice of $N$, we see that this entry is strictly greater than zero as well (see \eqref{L5-qtildek>0}).
\end{itemize}
Thus, the distribution $E(p)$ is full rank for a full rank $p$.
\end{remark}

\begin{lemma}
[Splitting of channel \cite{Brand_o_2015}]\label{lem:SplitChannel} Suppose
that a channel, for some fixed input probability distribution $u = (u_{1}
,\ldots, u_{\ell}, 0, \ldots, 0)$, where $u_{1} ,\ldots, u_{\ell}> 0$, outputs
the following probability distribution $u^{\prime}= (u^{\prime}_{1} ,\ldots,
u^{\prime}_{\ell}, 0, \ldots, 0)$. Moreover, suppose that $\Lambda(w) = w$
holds for some full rank distribution $w$. Then $\Lambda= \Lambda_{1}
\oplus\Lambda_{2}$, where $\Lambda_{1}$ acts on the first $\ell$ elements and
outputs to the first $\ell$ elements, and $\Lambda_{2}$ acts on the remaining
$n - \ell$ elements and outputs to the remaining $n - \ell$ elements.
\end{lemma}

\begin{proof}
Let $\mathcal{S}$ denote the set consisting of the first $\ell$ letters of the input alphabet (those
for which $u_{1},\ldots,u_{\ell}>0$). Consider the joint probability
distribution induced by the input distribution $w$ and the channel $\Lambda$.
Let $X$ be an indicator random variable, equal to zero if the channel input is
in $\mathcal{S}$ and equal to one if the channel input is in $\mathcal{S}^{c}%
$. Similarly, let $Y$ be an indicator random variable for $\mathcal{S}$ in the
same way for the channel output. Since the channel preserves $w$, it follows
that
\begin{align}
P(Y=0) &  =P(X=0),\\
P(Y=1) &  =P(X=1).
\end{align}
Moreover, since the channel transforms the input probability distribution $u$
to $u^{\prime}$, it follows that%
\begin{equation}
\sum_{j\in\mathcal{S}}\Lambda_{j|i}=1\ ~\forall i\in\mathcal{S}%
.\label{eq:condition-prob-mass-one-ch-1}%
\end{equation}
To see this, consider that%
\begin{equation}
u_{j}^{\prime}=\sum_{i}\Lambda_{j|i}u_{i}=\sum_{i\in\mathcal{S}}\Lambda_{j|i}u_{i},
\end{equation}
where the second equality follows because $u_{i}=0$ if $i\in\mathcal{S}^{c}$.
Now summing over $j\in\mathcal{S}$, we find that%
\begin{equation}
\sum_{j\in\mathcal{S}}u_{j}^{\prime}=1,
\end{equation}
again because $\mathcal{S}^{c}$ has zero probability mass for the output probability distribution $u'$. Then we conclude
that%
\begin{equation}
1=\sum_{j\in\mathcal{S}}\sum_{i\in\mathcal{S}}\Lambda_{j|i}u_{i}=\sum
_{i\in\mathcal{S}}u_{i}\sum_{j\in\mathcal{S}}\Lambda
(j|i).\label{eq:prob-mass-arg-ch-1}%
\end{equation}
From this condition, we can conclude that
\eqref{eq:condition-prob-mass-one-ch-1} holds. To prove this claim, we use the
method of contradiction. Suppose that it is not true, i.e., that there exists
some $i\in\mathcal{S}$ such that $\sum_{j\in\mathcal{S}}\Lambda_{j|i}<1$. Then
substituting back into the right-hand side of \eqref{eq:prob-mass-arg-ch-1}
and using the fact that $0<u_{i}<1$ for all $i\in\mathcal{S}$, it follows that
the sum in \eqref{eq:prob-mass-arg-ch-1} is strictly less than one, thus
giving a contradiction.

We now claim that the condition in \eqref{eq:condition-prob-mass-one-ch-1}
implies that $P\!\left(  Y=0|X=0\right)  =1$. To see this, consider that%
\begin{align}
P\!\left(  Y=0|X=0\right)    & =\frac{P\!\left(  Y=0,X=0\right)  }{P(X=0)}\\
& =\frac{\sum_{i,j\in\mathcal{S}}\Lambda_{j|i}w_{i}}{\sum_{i\in\mathcal{S}%
}w_{i}}\\
& =\frac{\sum_{i\in\mathcal{S}}w_{i}\sum_{j\in\mathcal{S}}\Lambda_{j|i}}%
{\sum_{i\in\mathcal{S}}w_{i}}\\
& =\frac{\sum_{i\in\mathcal{S}}w_{i}}{\sum_{i\in\mathcal{S}}w_{i}}\\
& =1.
\end{align}
Thus, we conclude that%
\begin{align}
\begin{split}
P(Y=0) &  =P(Y=0|X=0)P(X=0)\\
&  \quad\qquad+P(Y=0|X=1)P(X=1),
\end{split}
\\
&  =P(Y=0)+P(Y=0|X=1)P(X=1),
\end{align}
implying that
\begin{equation}
P(Y=0|X=1)P(X=1)=0.
\end{equation}
Since $w$ is full rank, it follows that $P(X=1)>0$. We then conclude that
\begin{equation}
P(Y=0|X=1)=0,
\end{equation}
from which we conclude that%
\begin{equation}
P(Y=1|X=1)=1.
\end{equation}
We can then rewrite this as follows:%
\begin{align}
1  & =P\!\left(  Y=1|X=1\right)  \nonumber\\
& =\frac{P(Y=1,X=1)}{P(X=1)}\\
& =\frac{\sum_{i,j\in\mathcal{S}^{c}}\Lambda_{j|i}w_{i}}{\sum_{i\in
\mathcal{S}^{c}}w_{i}}\\
& =\frac{\sum_{i\in\mathcal{S}^{c}}w_{i}\sum_{j\in\mathcal{S}^{c}}%
\Lambda_{j|i}}{\sum_{i^{\prime}\in\mathcal{S}^{c}}w_{i^{\prime}}}\\
& =\sum_{i\in\mathcal{S}^{c}}w_{i}^{\prime}\sum_{j\in\mathcal{S}^{c}}%
\Lambda_{j|i},
\end{align}
where $w_{i}^{\prime}$ denotes the following probability distribution on
$\mathcal{S}^{c}$:%
\begin{equation}
w_{i}^{\prime}:=w_{i}/\sum_{i^{\prime}\in\mathcal{S}^{c}}w_{i^{\prime}}.
\end{equation}
Now following the same reasoning used to arrive at
\eqref{eq:condition-prob-mass-one-ch-1}, we conclude that%
\begin{equation}
\sum_{j\in\mathcal{S}^{c}}\Lambda_{j|i}=1\ ~\forall i\in\mathcal{S}%
^{c}.\label{eq:condition-prob-mass-one-ch-2}%
\end{equation}
Combining \eqref{eq:condition-prob-mass-one-ch-1} and
\eqref{eq:condition-prob-mass-one-ch-2}, we conclude that first-group elements
are always mapped to first group and similarly for the second group. So the
channel can be written as a direct sum of two channels as indicated.
\end{proof}

\begin{lemma}[Continuity \cite{van_Erven_2014}] 
\label{lem:RenyiContinuity}
The relative entropy $\D{}{p}{q}$ is continuous in both arguments $p$ and $q$,  when $q$ has full rank. 
\end{lemma}

\begin{lemma}[Inclusion of support \cite{Renner2005SecurityOQ}] 
\label{lem:supp}
Let $\rho_{AB}$ be a density operator acting on $\mathcal{H}_A \otimes \mathcal{H}_B$. Then
\begin{equation}
	\supp{\rho_{AB}} \subseteq \supp{\rho_A} \otimes \supp{\rho_B}.
\end{equation}
\end{lemma}
\begin{proof}
First let us suppose that $\rho_{AB}$ is a pure state. Then $\rho_{AB} = \outprod{\psi}{\psi}$. Let $\ket{\psi} = \sum_{z\in\mathcal{Z}} \alpha_z \ket{\phi^z}\otimes\ket{\psi^z}$ be the Schmidt decomposition of $\ket{\psi}$.
Then
\begin{align}
\supp{\rho_{AB}} &= \{ \ket{\psi} \} \nonumber \\
&\subseteq \spn{\ket{\phi^z}_{z\in\mathcal{Z}}} \otimes \spn{\ket{\psi^z}_{z\in\mathcal{Z}}}.
\end{align}
Since $\spn{\ket{\phi^z}_{z\in\mathcal{Z}}} = \supp{\rho_A}$ and  $\spn{\ket{\psi^z}_{z\in\mathcal{Z}}} = \supp{\rho_B}$, we see that the lemma holds in this case.

For mixed states, let $\rho_{AB} = \sum_{x\in\mathcal{X}} \rho_{AB}^x$ be a decomposition into pure states for $\rho_{AB}$. Then
\begin{align}
& \!\!\! \supp{\rho_{AB}} \nonumber \\
&= \spn{\bigcup_{x\in\mathcal{X}} \supp{\rho_{AB}^x}}, \nonumber
\\
&\subseteq \spn{\bigcup_{x\in\mathcal{X}} \supp{\rho_A^x} \otimes \supp{\rho_B^x}}, \nonumber
\\ 
\begin{split}
&\subseteq \left( \spn{\bigcup_{x\in\mathcal{X}} \supp{\rho_A^x}} \right) \otimes \\
              &\qquad \qquad \left( \spn{\bigcup_{x\in\mathcal{X}} \supp{\rho_B^x}} \right),
\end{split} \nonumber 
\\
	&= \supp{\rho_A} \otimes \supp{\rho_B}.
\end{align}
This concludes the proof.
\end{proof}

\begin{lemma}[Superadditivity of Relative Entropy \cite{Capel_2018}] 
\label{lem:SuperS1}
Let $\mathcal{H}_{AB} = \mathcal{H}_A \otimes \mathcal{H}_B$ be a bipartite Hilbert space, and let $\rho_{AB}$,   $\sigma_A$, and $\sigma_B$ be density operators. Then
\begin{align}
& \!\!\! \SD{}{\rho_{AB}}{\sigma_A \otimes \sigma_B} \nonumber \\
&= \SD{}{\rho_{AB}}{\rho_A \otimes \rho_B} + \SD{}{\rho_A}{\sigma_A} + \SD{}{\rho_B}{\sigma_B}.
\end{align}
Thus, \begin{equation}
 \SD{}{\rho_{AB}}{\sigma_A \otimes \sigma_B} 
\geq   \SD{}{\rho_A}{\sigma_A} + \SD{}{\rho_B}{\sigma_B},
\end{equation}
and equality holds if and only if $\rho_{AB} = \rho_A \otimes \rho_B$.
\end{lemma}
\begin{proof}
Using the definition of quantum relative entropy, we find that
	\begin{align}
		&\!\!\! \SD{}{\rho_{AB}}{\sigma_A \otimes \sigma_B} \nonumber \\
		&= \Tr \left[ \rho_{AB} \left(\ln\rho_{AB} -\ln\sigma_A \otimes \sigma_B \right) \right] \nonumber
		\\
		\begin{split}
			 &= \Tr \left[ \rho_{AB} \left(\ln\rho_{AB} -\ln\rho_A \otimes \rho_B \right)\right] \nonumber \\
             &\qquad + \Tr \left[\rho_{AB}\left(\ln\rho_A \otimes \rho_B -\ln\sigma_A \otimes \sigma_B \right) \right]
		\end{split}	\nonumber
		\\
		\begin{split}		
			&= \SD{}{\rho_{AB}}{\rho_A \otimes \rho_B} \\
			&\qquad + \SD{}{\rho_A \otimes \rho_B}{\sigma_A \otimes \sigma_B}
		\end{split} \nonumber
		\\
		\begin{split}
		&= \SD{}{\rho_{AB}}{\rho_A \otimes \rho_B} + \SD{}{\rho_A}{\sigma_A} + \SD{}{\rho_B}{\sigma_B}.
		\end{split}
	\end{align}
Since $\SD{}{\rho_{AB}}{\rho_A \otimes \rho_B} \geq 0$ (see \eqref{QuaRenDiv>0}), we conclude that ${\SD{}{\rho_{AB}}{\sigma_A \otimes \sigma_B} \ge \SD{}{\rho_A}{\sigma_A} + \SD{}{\rho_B}{\sigma_B}}$. The conclusion about equality holding is a direct consequence of $\SD{}{\rho_{AB}}{\rho_A \otimes \rho_B} = 0$ if and only if $\rho_{AB} = \rho_A \otimes \rho_B$.
\end{proof}

\begin{lemma}[Superadditivity of $D_0$]
\label{lem:SuperS0}
Let $\mathcal{H}_{AB} = \mathcal{H}_A \otimes \mathcal{H}_B$ be a bipartite Hilbert space, and let $\rho_{AB}$, $\sigma_A$, and  $ \sigma_B$ be density operators. Then
\begin{equation}
\SD{0}{\rho_{AB}}{\sigma_A \otimes \sigma_B} \geq \SD{0}{\rho_A}{\sigma_A} + \SD{0}{\rho_B}{\sigma_B}.
\end{equation}
\end{lemma}
\begin{proof}
From Lemma \ref{lem:supp}, we see that $\supp{\rho_{AB}} \subseteq \supp{\rho_A} \otimes \supp{\rho_B}$ and thus the following operator inequality holds
\begin{equation}
\Pi_{\rho_{AB}} \leq \Pi_{\rho_A} \otimes \Pi_{\rho_B},
\end{equation}
where $\Pi_{\omega}$ denotes the projection onto the support of the state $\omega$.
Then,
\begin{align}
	& \!\!\!\! \SD{0}{\rho_{AB}}{\sigma_A \otimes \sigma_B} \nonumber \\
	&= -\ln\left( \Tr\left[ \Pi_{\rho_{AB}} (\sigma_A \otimes \sigma_B) \right] \right) \nonumber \\
	&\geq -\ln\left( \Tr\left[ (\Pi_{\rho_A} \otimes \Pi_{\rho_B}) (\sigma_A \otimes \sigma_B) \right] \right) \nonumber \\
	&= -\ln\left( \Tr\left[ (\Pi_{\rho_A} \sigma_A \otimes \Pi_{\rho_B} \sigma_B)\right] \right) \nonumber \\
	&= -\ln\left( \Tr\left[\Pi_{\rho_A} \sigma_A\right] \Tr\left[\Pi_{\rho_B} \sigma_B\right] \right) \nonumber \\
	&= -\ln\left( \Tr\left[\Pi_{\rho_A} \sigma_A\right] \right)-\ln\left( \Tr\left[\Pi_{\rho_B} \sigma_B\right] \right) \nonumber \\
	&= \SD{0}{\rho_A}{\sigma_A} + \SD{0}{\rho_B}{\sigma_B}.
\end{align}
This concludes the proof.
\end{proof}

Lemmas \ref{lem:SuperS1} and \ref{lem:SuperS0} are results based on  quantum divergences. However, any result that holds for the quantum case also holds for the classical case. This is because we can plug in commuting quantum states as a special case and these states are  quasi-classical.

\begin{lemma}[M\"uller \cite{M_ller_2018}]
\label{lem:Muller}
Let $p_A, p'_A \in \mathbb{R}^k$ be probability distributions with $p_A^\downarrow \neq p_A'^\downarrow$. Then there exists a probability distribution $q_B$ and an extension $p'_{AB}$ with marginals $p'_A$ and $q_B$ such that
\begin{equation}
p_A \otimes q_B \succ p'_{AB}
\end{equation}
if and only if $H_0(p_A) \leq H_0(p'_A)$ and $H(p_A) < H(p'_A)$. Moreover, if these inequalities are satisfied, then for all $\varepsilon > 0$, we can choose system $B$ and $p'_{AB}$ such that
\begin{equation}
\D{}{p'_{AB}}{p'_A \otimes q_B} < \varepsilon.
\end{equation}
\end{lemma}

\section{Proof of Theorem 1}
\label{sec:proofThm1}

\begin{proof}
Since $q$ and $q'$ have rational entries, without loss of generality, we pick $q = \left( \frac{d_1}{N}, \ldots, \frac{d_k}{N}\right)$ and $q' = \left( \frac{d'_1}{N}, \ldots, \frac{d'_k}{N}\right)$ for sets $\{d_i\}_i$ and $\{d'_i\}_i$ of natural numbers, such that
\begin{equation}
    \sum\limits_{i=1}^k d_i = \sum\limits_{i=1}^k d'_i = N,
    \label{eq:cond-on-N-spec}
\end{equation}
respectively.

We begin with statement 1 implies statement 2; i.e., we suppose that $\D{}{p}{q} > \D{}{p'}{q'}$ and $\D{0}{p}{q} \geq \D{0}{p'}{q'}$, and prove the existence of a classical channel $\Lambda$, a probability distribution $r$, and a joint distribution $t'$ that satisfy the conditions of statement 2.

Using Lemma \ref{lem:PreserveRenyi} and statement 1, we conclude that
\begin{equation}
\begin{aligned}
	\D{}{\Gamma_d(p)}{\eta_N} &> \D{}{\Gamma_{d'}(p')}{\eta_N}, \\
	\D{0}{\Gamma_d(p)}{\eta_N} &\geq \D{0}{\Gamma_{d'}(p')}{\eta_N}.
\end{aligned}
\end{equation}
Define $\tilde{p} \equiv \Gamma_d(p)$ and $\tilde{p}' \equiv \Gamma_{d'}(p')$. Thus,
\begin{equation}
\label{T1-Step1}
\begin{aligned}
	\D{}{\tilde{p}}{\eta_N} &> \D{}{\tilde{p}'}{\eta_N}, \\
	\D{0}{\tilde{p}}{\eta_N} &\geq \D{0}{\tilde{p}'}{\eta_N}.
\end{aligned}
\end{equation}

Recalling \eqref{eq:renyi-div-to-unif-relate-ent}, the following equality holds  for all probability distributions~$\tilde{p}$ and $\alpha \in \{0, 1\}$:
\begin{equation}
\label{T1-RenEntDef}
    \D{\alpha}{\tilde{p}}{\eta_N} =\ln(N) - H_\alpha(\tilde{p}).
\end{equation}
Thus, from \eqref{T1-Step1} and \eqref{T1-RenEntDef}, we conclude that
\begin{equation}
\label{T1-RenEnt1}
\begin{aligned}
	H(\tilde{p}) &< H(\tilde{p}'), \\
	H_0(\tilde{p}) &\leq H_0(\tilde{p}').
\end{aligned}
\end{equation}
Furthermore, we conclude that
\begin{equation}
    \tilde{p}^{\downarrow} \neq \tilde{p}'^{\downarrow}, \label{eq:spectra-diff-assump}
\end{equation}
which follows from the structure of the channels $\Gamma_d$ and $\Gamma_{d'}$, as well as the assumption that the relative spectra of the pairs $(p,q)$ and $(p',q')$ are different.

Using \eqref{T1-RenEnt1} and \eqref{eq:spectra-diff-assump}, we apply Lemma~\ref{lem:Muller} to conclude that there exists a probability distribution $r$ and an extension $v'$ with marginals $\tilde{p}'$ and $r$, such that
\begin{equation}
\begin{aligned}
	\tilde{p} \otimes r &\succ v', \\
	\D{}{v'}{\tilde{p}' \otimes r} &< \gamma,
\end{aligned}
\label{T1-MajStep}
\end{equation}
for all $\gamma > 0$. In other words, there exists a  doubly stochastic channel $\Phi$ (see Lemma \ref{lem:mjrz}) such that
\begin{equation}
	\Phi(\tilde{p} \otimes r) = v'.
\end{equation}

Using Lemma~\ref{lem:SplitChannel}, we conclude that $r$ can be considered, without loss of generality, to be of full rank. Concretely, if $r$ is not full rank, then define $u = \tilde{p} \otimes r$, $u' = v'$, and $w = \eta_N \otimes \eta$. Using Lemma~\ref{lem:SplitChannel}, $\Phi$ can be split into $\Phi_1 \oplus \Phi_2$. If $r$ is of full rank, then $\Phi = \Phi_1$.

Finally consider the following channel:
\begin{equation}
	\Lambda := (\Gamma_{d'}^* \otimes \operatorname{id}) \circ \Phi_1 \circ (\Gamma_d \otimes \operatorname{id}).
\end{equation}
Since this is a composition of classical channels, the overall map is a classical channel.\\

\textbf{Proving 2a)}
\begin{align}
	\Lambda(p \otimes r) &= [(\Gamma_{d'}^* \otimes \operatorname{id}) \circ \Phi_1 \circ (\Gamma_d \otimes \operatorname{id})](p \otimes r), \nonumber \\
	&= [(\Gamma_{d'}^* \otimes \operatorname{id}) \circ \Phi_1](\tilde{p} \otimes r), \nonumber \\
	&= (\Gamma_{d'}^* \otimes \operatorname{id})(v').
\end{align}

We now define ${t' = (\Gamma_{d'}^* \otimes \operatorname{id})(v')}$, and  we observe that the marginals of $t'$ are $\Gamma_{d'}^*(\tilde{p}') = p'$ and $r$. \\

\textbf{Proving 2b)}\\

We need to show that $\Lambda(q \otimes \eta) = q' \otimes \eta$.  Consider that
\begin{align}
	\Lambda(q \otimes \eta) &= [(\Gamma_{d'}^* \otimes \operatorname{id}) \circ \Phi_1 \circ (\Gamma_d \otimes \operatorname{id})](q \otimes \eta), \nonumber \\
	&= [(\Gamma_{d'}^* \otimes \operatorname{id}) \circ \Phi_1](\eta_N \otimes \eta), \nonumber \\
	&= (\Gamma_{d'}^* \otimes \operatorname{id})(\eta_N \otimes \eta), \nonumber \\	
	&= q' \otimes \eta.
\end{align}
where the second equality follows from Lemma \ref{lem:embedRational}. Thus, ${\Lambda(q \otimes \eta) = q' \otimes \eta}$.\\

\textbf{Proving 2c)}
\begin{align}
	\gamma &> \D{}{v'}{\tilde{p}'\otimes r}, \nonumber \\
	&\geq \D{}{t'}{p' \otimes r},
\end{align}
where the first inequality follows from \eqref{T1-MajStep} and the second inequality follows from the data-processing inequality in \eqref{RenDivDatProc}, using the channel $\Gamma^{*}_{d'} \otimes \operatorname{id}$. Thus, $\D{}{t'}{p' \otimes r} \leq \gamma$.

This completes the proof that statement 1 implies statement 2.\\

We now establish the reverse implication; i.e., we suppose that there exists a classical channel $\Lambda$ satisfying the conditions of statement 2. Let $\alpha \in \{0, 1\} $. Then,
\begin{align}
    \label{T1-backDer}
	\D{\alpha}{p}{q} + \D{\alpha}{r}{\eta} &= \D{\alpha}{p \otimes r}{q \otimes \eta} \nonumber \\
	&\geq \D{\alpha}{\Lambda(p \otimes r)}{\Lambda(q \otimes \eta)}, \nonumber \\
	&= \D{\alpha}{t'}{q' \otimes \eta}, \nonumber \\
	&\geq \D{\alpha}{p'}{q'} + \D{\alpha}{r}{\eta},
\end{align}
where the first equality follows from the additivity of the \RD in \eqref{RenDivAdd}, the first inequality follows from the data-processing inequality in \eqref{RenDivDatProc}, and the last inequality follows from Lemmas~\ref{lem:SuperS1} and \ref{lem:SuperS0} for $\alpha = 1$ and $\alpha = 0$, respectively.

Since the support of $r$ is contained in the support of $ \eta$, it follows that $\D{\alpha}{r}{\eta}$ is finite and can be subtracted, and we conclude that the following inequality holds for $\alpha \in \{0,1\}$:
\begin{equation}
	\D{\alpha}{p}{q} \geq \D{\alpha}{p'}{q'}.
\end{equation}

To conclude that the inequality 
\begin{equation}
\D{}{p}{q} \geq \D{}{p'}{q'}
\label{eq:app-rel-ent-ineq-thm-1}
\end{equation}
 is actually strict, consider that the equality $\D{}{p}{q} = \D{}{p'}{q'}$ is equivalent to equality in \eqref{T1-backDer}. Saturating the last inequality of \eqref{T1-backDer} is then equivalent to $t' = p' \otimes r$. This in turn implies that the forward channel $\Lambda$ realizes the following transformation:
\begin{align}
    \Lambda(p \otimes r) & = p' \otimes r, \\
    \Lambda(q \otimes \eta) & = q' \otimes \eta .
\end{align}
Now defining $\tilde{p} = \Gamma_d(p)$ and $\tilde{p}' = \Gamma_{d'}(p')$ for $N$ satisfying \eqref{eq:cond-on-N-spec}, and recalling that $\eta_N = \Gamma_d(q)$ and $\eta_N =   \Gamma_{d'}(q')$, we conclude that
\begin{align}
    [(\Gamma_{d'} \otimes \operatorname{id})\circ \Lambda\circ (\Gamma^*_d \otimes \operatorname{id})](\tilde{p} \otimes r) & = \tilde{p}' \otimes r, \\
    [(\Gamma_{d'} \otimes \operatorname{id})\circ \Lambda\circ (\Gamma^*_d \otimes \operatorname{id})](\eta_N \otimes \eta) & = \eta_N \otimes \eta .
\end{align}
So this means that $(\Gamma_{d'} \otimes \operatorname{id})\circ \Lambda\circ (\Gamma^*_d \otimes \operatorname{id})$ is a doubly stochastic channel, and in turn that $ \tilde{p}  \succ_T \tilde{p}' $ for $\tilde{p}^\downarrow \neq \tilde{p}'^\downarrow$ (the latter following from the assumption of differing relative spectra). By applying the main result of \cite{klimesh2007inequalities}, it follows that $H(\tilde{p}) < H(\tilde{p}')$, which in turn is equivalent to $\D{}{p}{q} > \D{}{p'}{q'}$ by applying \eqref{T1-RenEntDef}. This contradicts the possibility of equality in \eqref{eq:app-rel-ent-ineq-thm-1}, and so we  conclude that only a strict inequality can hold.
\end{proof}

\section{Proof of Theorem 2}
\label{sec:proofThm2}

\begin{proof}

Some aspects of the proof given below follow the proof of \cite[Theorem~17]{Brand_o_2015} closely.

We begin with statement 1 implies statement 2; i.e., we suppose that $\D{}{p}{q} \geq \D{}{p'}{q'}$ and prove the existence of a classical channel $\Lambda$, probability distributions $p'_\varepsilon$ and $r$, and a joint distribution $t'_\varepsilon$ that satisfy the conditions of statement 2.

As a first check, we can determine by means of an efficient algorithm (see, e.g., \cite{Dahl99} or  \cite{Renes_2016}) if the pair $(p,q)$ relatively majorizes the pair $(p',q')$. If this is the case, then the claim trivially holds without catalysis because there exists a classical channel that takes $p$ to $p'$ and $q$ to $q'$. 

In the more general case of interest for this paper, it may not be the case that the pair $(p,q)$ relatively majorizes the pair $(p',q')$. Furthermore, it is not necessarily the case that  $q$ and $q'$  have rational entries. However, the set of probability distributions with rational entries is dense in the set of all distributions. Thus, using Lemma~\ref{lem:ConvertRational}, we can always pick probability distributions $q_d$ and $q'_{d'}$ (defined below) that are arbitrarily close to $q$ and $q'$ respectively. 

Without loss of generality, we pick $q_d = \left( \frac{d_1}{N}, \ldots, \frac{d_k}{N}\right)$ and  $q'_{d'} = \left( \frac{d'_1}{N}, \ldots, \frac{d'_k}{N}\right)$ for integers $\{d_i\}$ and $\{d'_i\}$, such that $\sum\limits_{i=1}^k d_i = \sum\limits_{i=1}^k d'_i = N$, respectively, with $N$ sufficiently large as needed for the application of Lemma~\ref{lem:ConvertRational}.

Applying Lemma~\ref{lem:ConvertRational}, we conclude that there exist classical channels $E$ and $E'$ such that $E(q) = q_d$ and $E'(q') = q'_{d'}$,
\begin{align}
    \label{T2-DefE}
	\frac{1}{2}\left \Vert q_d - q \right\Vert_1 & \leq \varepsilon_1,  \nonumber \\
	 \frac{1}{2}\left \Vert E(p) - p \right\Vert_1 & \leq O(\sqrt{\varepsilon_1}),
\end{align}
for all probability distributions $p$, and 
\begin{align}
    \label{T2-DefE'}
	\frac{1}{2}\left\Vert q'_{d'} - q' \right\Vert_1 & \leq \varepsilon_2, \nonumber \\
	\frac{1}{2}\left\Vert E'(p') - p' \right\Vert_1 & \leq O(\sqrt{\varepsilon_2}),
\end{align}
for all probability distributions $p'$.

Define the reversal channel of $E'$ as $E'^{*}$ (see Definition~\ref{rmrk:Reversal}). Then
\begin{equation}
	E'^{*}(q'_{d'}) = q'.
\end{equation}
and applying Lemma~\ref{lem:ConvertRational}, it follows that
\begin{equation}
	\frac{1}{2}\left\Vert E'^{*}(p') - p'\right\Vert_1 \leq O(\sqrt{\varepsilon_2}),
	\label{eq:reversal-ch-final-proof}
\end{equation}
for all probability distributions $p'$.
Now define 
\begin{equation}
    \label{T2-Defp''}
    p'' = (1-\delta)p' + \delta q'
\end{equation}
for $0<\delta<1$. Note that 
\begin{equation}
\frac{1}{2} \Vert p'' - p' \Vert_1 \leq \delta   
\label{eq:delta-ineq-p-double-p-single}
\end{equation}
 and $p''$ is full rank. From Lemma \ref{lem:MixRenyi} and statement 1, we conclude that
\begin{equation}
    \label{T2-RelEntp''}
    \D{}{p}{q} > \D{}{p''}{q'}.
\end{equation}

Lemma~\ref{lem:RenyiContinuity} states that $\D{}{p}{q}$ is continuous in both arguments $p$ and $q$, when $q$ is full rank. Thus it follows  that 
\begin{align}
    \label{T2-channelE}
    |\D{}{p}{q} - \D{}{E(p)}{q_d}| \leq f_1(\varepsilon_1),
\end{align}
for some function $f_1(\varepsilon_1)$ with the property that $\lim_{\varepsilon_1 \to 0} f_1(\varepsilon_1) =0 $.
Similarly, it follows that
\begin{align}
    \label{T2-channelE'}
    |\D{}{p''}{q'} - \D{}{E'(p'')}{q'_{d'}}| \leq f_2(\varepsilon_2),
\end{align}
for some function $f_2(\varepsilon_1)$ with the property that $\lim_{\varepsilon_1 \to 0} f_2(\varepsilon_1) =0 $.
Thus, using \eqref{T2-RelEntp''}, \eqref{T2-channelE}, and \eqref{T2-channelE'}, and taking $\varepsilon_1$ and $\varepsilon_2$ sufficiently small by taking $N$ sufficiently large, so that
\begin{equation}
    \D{}{p}{q} - \D{}{p''}{q'} > f(\varepsilon_1) + f(\varepsilon_2),
\end{equation}
we conclude that
\begin{equation}
\label{T2-Step2}
    \D{}{E(p)}{q_d} > \D{}{E'(p'')}{q'_{d'}}.
\end{equation}

Then using Lemma \ref{lem:PreserveRenyi} and \eqref{T2-Step2}, we conclude that
\begin{equation}
	\D{}{\Gamma_d(E(p))}{\eta_N} > \D{}{\Gamma_{d'}(E'(p''))}{\eta_N}.
\end{equation}
Define
\begin{align}
\tilde{E}(p) & \equiv \Gamma_d(E(p)), \\  \tilde{E}'(p'') & \equiv \Gamma_{d'}(E'(p'')). 
\end{align}
Thus,
\begin{equation}
\label{T2-RenDevUni}
    \D{}{\tilde{E}(p)}{\eta_N} > \D{}{\tilde{E}'(p'')}{\eta_N}.
\end{equation}

Applying \eqref{T2-RenDevUni} and \eqref{eq:renyi-div-to-unif-relate-ent}, we conclude that
\begin{equation}
\label{T2-RenEnt1}
    H(\tilde{E}(p)) < H(\tilde{E}'(p'')).
\end{equation}
Since $p''$ is full rank, it follows from Remark \ref{rmrk:FullRank} that $E'(p'')$ is full rank also, so that $\tilde{E}'(p'')$ is full rank as well. From the definition of $H_0$ (see Subsection \ref{subsec:RE}), it follows that
\begin{equation}
\label{T2-RenEnt0}
    H_0(\tilde{E}(p)) \leq H_0(\tilde{E}'(p'')).
\end{equation}
Due to fact that \eqref{T2-RenEnt1} holds, we can also conclude that
\begin{equation}
\tilde{E}(p)^\downarrow \neq \tilde{E}'(p'')^\downarrow. \label{eq:differing-spectra-final-final}    
\end{equation}
 For if we had $\tilde{E}(p)^\downarrow = \tilde{E}'(p'')^\downarrow$, then necessarily $H(\tilde{E}(p)) = H(\tilde{E}'(p''))$.

Since \eqref{T2-RenEnt1}, \eqref{T2-RenEnt0}, and \eqref{eq:differing-spectra-final-final} hold, we can invoke Lemma~\ref{lem:Muller} to conclude that there exist a probability distribution $r$ and a joint distribution  $v''$ with marginals $\tilde{E}'(p'')$ and $r$ such that
\begin{equation}
\begin{gathered}
	\tilde{E}(p) \otimes r \succ v'',  \\
	\D{}{v''}{\tilde{E}'(p'') \otimes r} < \gamma,
\end{gathered}
\label{T2-MajStep}
\end{equation}
for all $\gamma > 0$. Thus, there exists a doubly stochastic channel $\Phi$ (see Lemma \ref{lem:mjrz}) such that
\begin{equation}
	\Phi(\tilde{E}(p) \otimes r) = v''.
\end{equation}

Applying Lemma~\ref{lem:SplitChannel}, it follows that $r$ can be considered, without loss of generality, to be of full rank. Concretely, if $r$ is not full rank, then define $u = \tilde{E}(p) \otimes r$, $u' = v''$ and $w = \eta_N \otimes \eta$. Using Lemma~\ref{lem:SplitChannel}, $\Phi$ can be split into $\Phi_1 \oplus \Phi_2$. If $r$ is of full rank, then $\Phi = \Phi_1$.

Finally, we define the following channel:
\begin{equation}
	\Lambda = (E'^{*} \otimes \operatorname{id}) \circ (\Gamma_{d'}^* \otimes \operatorname{id}) \circ \Phi_1 \circ (\Gamma_d \otimes \operatorname{id}) \circ (E \otimes \operatorname{id}).
\end{equation}
Since this is a composition of classical channels, the overall map is a classical channel.\\

\textbf{Proving 2a) and 2c)}
\begin{align} 
	&\!\!\! \Lambda(p \otimes r) \nonumber \\
	&= [(E'^{*} \otimes \operatorname{id}) \circ (\Gamma_{d'}^* \otimes \operatorname{id}) \circ \Phi_1 \nonumber \\
	& \qquad \circ (\Gamma_d \otimes \operatorname{id}) \circ (E \otimes \operatorname{id})](p \otimes r), \nonumber \\
	&= [(E'^{*} \otimes \operatorname{id}) \circ (\Gamma_{d'}^* \otimes \operatorname{id}) \circ \Phi_1 \circ (\Gamma_d \otimes \operatorname{id})](E(p) \otimes r), \nonumber \\
	&= [(E'^{*} \otimes \operatorname{id}) \circ (\Gamma_{d'}^* \otimes \operatorname{id}) \circ \Phi_1](\tilde{E}(p) \otimes r), \nonumber \\
	&= [(E'^{*} \otimes \operatorname{id}) \circ (\Gamma_{d'}^* \otimes \operatorname{id})]v''
\end{align}

We now define ${t'_\varepsilon = [(E'^{*} \otimes \operatorname{id}) \circ (\Gamma_{d'}^* \otimes \operatorname{id})]v''}$. Observe that the marginals of $t'_\varepsilon$ are $p'_\varepsilon = E'^{*}(\Gamma_{d'}^*(\tilde{E}'(p'')))$ and $r$. 

Now consider the following chain of inequalities:
\begin{align}
    &\!\!\! \frac{1}{2}\left\Vert p'_\varepsilon - p' \right\Vert_1 \nonumber \\
    &= \frac{1}{2}\left\Vert E'^{*}(\Gamma_{d'}^*(\tilde{E}'(p''))) - p' \right \Vert_1 \nonumber \\
	\begin{split}
		&\leq \frac{1}{2}\left\Vert E'^{*}(\Gamma_{d'}^*(\tilde{E}'(p''))) - \Gamma_{d'}^*(\tilde{E}'(p'')) \right\Vert_1 \\
		&\qquad\qquad + \frac{1}{2}\left\Vert \Gamma_{d'}^*(\tilde{E}'(p'')) -p' \right\Vert_1
	\end{split} \nonumber 
	\\
	&\leq O(\sqrt{\varepsilon_2}) + \frac{1}{2}\left\Vert \Gamma_{d'}^*(\Gamma_{d'}(E'(p''))) - p' \right\Vert_1	 \nonumber \\
	&= O(\sqrt{\varepsilon_2}) + \frac{1}{2}\left\Vert E'(p'') - p' \right\Vert_1 \nonumber \\
	&\leq O(\sqrt{\varepsilon_2}) + \frac{1}{2}\left\Vert E'(p'') - p'' \right\Vert_1 + \frac{1}{2}\left\Vert p'' - p' \right\Vert_1 \nonumber \\
	&\leq O(\sqrt{\varepsilon_2}) + O(\sqrt{\varepsilon_1}) +  \delta ,
	\label{eq:final-final-last-line}
\end{align}
where the first inequality follows from the triangle inequality, the second inequality follows from \eqref{eq:reversal-ch-final-proof}, the second equality follows from \eqref{eq:left-inv-ch-embed}, the third inequality follows from the triangle inequality, and the last inequality follows from \eqref{T2-DefE'} and~\eqref{eq:delta-ineq-p-double-p-single}.

Since $\varepsilon_1$, $\varepsilon_2$ and $\delta$ can be made arbitrarily small, we can set the right-hand side of \eqref{eq:final-final-last-line} to $\varepsilon$. Thus $\Lambda(p \otimes r) = t'_\varepsilon$ with marginals $p'_\varepsilon$ and $r$, and the following inequality holds $\frac{1}{2}\left\Vert p'_\varepsilon - p' \right\Vert_1 \leq \varepsilon$.\\

\textbf{Proving 2b)} \\

We need to show that $\Lambda(q \otimes \eta) = q' \otimes \eta$. Then, 
\begin{align}
	& \Lambda(q \otimes \eta) \nonumber \\
	&= [(E'^{*} \otimes \operatorname{id}) \circ (\Gamma_{d'}^* \otimes \operatorname{id}) \circ \Phi_1 \nonumber \\
	& \qquad \circ (\Gamma_d \otimes \operatorname{id}) \circ (E \otimes \operatorname{id})](q \otimes \eta), \nonumber \\
	&= [(E'^{*} \otimes \operatorname{id}) \circ (\Gamma_{d'}^* \otimes \operatorname{id}) \circ \Phi_1 \circ (\Gamma_d \otimes \operatorname{id})](q_d \otimes \eta), \nonumber \\
	&= [(E'^{*} \otimes \operatorname{id}) \circ (\Gamma_{d'}^* \otimes \operatorname{id}) \circ \Phi_1](\eta_N \otimes \eta), \nonumber \\
	&= [(E'^{*} \otimes \operatorname{id}) \circ (\Gamma_{d'}^* \otimes \operatorname{id})](\eta_N \otimes \eta), \nonumber \\
	&= [(E'^{*} \otimes \operatorname{id})](q'_{d'} \otimes \eta), \nonumber \\
	&= q' \otimes \eta,
\end{align}
where the third equality follows because $\Phi_1$ is a doubly stochastic channel. Thus $\Lambda(q\otimes \eta) = q' \otimes \eta$.\\

\textbf{Proving 2d)}
\begin{align}
	\gamma &> \D{}{v''}{\tilde{E}''(p') \otimes r}, \nonumber \\
	&\geq \D{}{t'_\varepsilon}{p'_\varepsilon \otimes r},
\end{align}
where the first inequality follows from \eqref{T2-MajStep} and the second inequality follows from the data-processing inequality \eqref{RenDivDatProc} using the channel $(E'^{*} \otimes \operatorname{id}) \circ (\Gamma^*_{d'} \otimes \operatorname{id})$. Thus, $\D{}{t'_\varepsilon}{p'_\varepsilon \otimes r} \leq \gamma$.
This completes the proof that statement 1 implies statement 2.

We now look at the reverse direction; i.e., fix $\varepsilon \in (0,1)$ and $\gamma > 0 $ and suppose that there are probability distributions $r$  and $p_{\varepsilon}'$, a joint distribution $t_{\varepsilon}'$,  and   a classical channel $\Lambda$ satisfying the stated conditions. Then consider that
\begin{align}
	\D{}{p}{q} + \D{}{r}{\eta} &= \D{}{p \otimes r}{q \otimes \eta}, \nonumber \\
	&\geq \D{}{\Lambda(p \otimes r)}{\Lambda(q \otimes \eta)}, \nonumber \\
	&= \D{}{t'_\varepsilon}{q' \otimes \eta}, \nonumber \\
	&\geq \D{}{p'_\varepsilon}{q'} + \D{}{r}{\eta},
\end{align}
where the first equality follows from additivity of relative entropy \eqref{RenDivAdd}, the first inequality follows from the data-processing inequality \eqref{RenDivDatProc}, and the last inequality follows from Lemma~\ref{lem:SuperS1}.
Since $\eta$ is full rank, $\D{}{r}{s}$ is finite and can be subtracted, so that the following inequality holds
\begin{equation}
	\D{}{p}{q} \geq \D{}{p'_\varepsilon}{q'}.
\end{equation}
Since this inequality holds for all $\varepsilon \in (0,1)$, we can apply 
 Lemma \ref{lem:RenyiContinuity} and take the limit as $\varepsilon \to 0$ to conclude that
\begin{equation}
	\D{}{p}{q} \geq \D{}{p'}{q'}.
\end{equation}
This concludes the proof.
\end{proof}

\section{Proof of Corollary 1}
\label{sec:proofCor1}

\begin{proof} 
Simply pick $q = q' = \eta$ and apply Theorem \ref{Theorem2}. Since $\Lambda$ preserves the uniform distribution, it is a unital classical channel.
\end{proof}

\end{document}